\documentclass[a4paper]{article}
\usepackage{a4wide}
\usepackage{named}

\usepackage{times}
\usepackage{amsmath,amsfonts,amssymb}
\usepackage{hhline}
\usepackage{pifont} % for \ding
\usepackage{tikz}
\usetikzlibrary{positioning,arrows,fit,decorations.pathmorphing}
\usetikzlibrary{patterns}
\usetikzlibrary{calc}
\usepackage{url}
\usepackage[only bigsqcap]{stmaryrd} %% for \bigsqcap

\newtheorem{theorem}{Theorem}
\newtheorem{lemma}[theorem]{Lemma}
\newtheorem{example}[theorem]{Example}

\newcommand{\qedsymbol}{\ding{113}}
\newenvironment{proof}{\par\noindent\textbf{Proof.}}{\mbox{}\hfill\qedsymbol\par\bigskip}

\newcommand{\OWLQL}{\textsl{OWL\,2\,QL}}
\newcommand{\TQL}{\textsl{TQL}}

\newcommand{\DL}{\textsl{DL-Lite}}
\newcommand{\DLh}{{\ensuremath{\DL_\textit{horn}^{\smash{\mathcal{H}}}}}}
\newcommand{\q}{\boldsymbol q}

\newcommand{\coNP}{\textsc{coNP}}
\newcommand{\NP}{\textsc{NP}}

\newcommand{\ind}{\mathop{\mathsf{ind}}}
\newcommand{\ext}{\mathsf{ext}}
\newcommand{\tw}{\mathsf{w}}
\newcommand{\op}{\mathsf{cl}}

\newcommand{\type}{\boldsymbol{t}}

\newcommand{\A}{\mathcal{A}}
\newcommand{\T}{\mathcal{T}}
\newcommand{\I}{\mathcal{I}}
\newcommand{\K}{\mathcal{K}}
\newcommand{\C}{\mathcal{C}}
\newcommand{\Smc}{\mathcal{S}}

\newcommand{\g}{\mathfrak{g}}

\newcommand{\tc}{\mathsf{tem}}

\newcommand{\nxt}{{\ensuremath\raisebox{0.25ex}{\text{\scriptsize$\bigcirc$}}}}

\newcommand{\FT}{\scriptscriptstyle\! F}
\newcommand{\PT}{\scriptscriptstyle\! P}

\newcommand{\nm}[1]{\textit{#1}}
\newcommand{\roleT}{\mathsf{R}_\T} %% you can replace back to r^\pm(\T) 
\newcommand{\CKTR}{\mathcal{C}_{\smash{\K_{\T\!\!,R}}}}

\title{Temporal Description Logic for Ontology-Based Data Access (Extended Version)}
\author{Alessandro Artale$^1$\!, Roman Kontchakov$^2$\!, Frank Wolter$^3$ and Michael Zakharyaschev$^2$
\\
{\parbox[t]{70mm}{\centering $^1$Faculty of Computer Science\\
Free University of Bozen-Bolzano, Italy\\
$^3$ Department of Computer Science,\\ University of Liverpool, U.K.}}\hspace*{7mm}
{\parbox[t]{70mm}{\centering $^2$Department of Computer Science and Information Systems\\ Birkbeck, University of London, U.K.}}
}
\begin{document}
\maketitle
\begin{abstract}
Our aim is to investigate ontology-based data access over temporal data with validity time and ontologies capable of  temporal conceptual modelling. To this end, we design a temporal description logic, \TQL, that extends the standard ontology language \OWLQL, provides basic means for temporal conceptual modelling and ensures first-order rewritability of conjunctive queries for suitably defined data instances with validity time.
%In ontology-based data access (OBDA), the most important 
%technique for answering queries is to rewrite the ontology and the query into a single first-order query that can be evaluated directly over the data. The current W3C standard for OBDA is OWL 2 QL which is based on the DL-Lite family of description logics. The aim of this paper is to extend OBDA to data with validity time and ontologies that admit temporal conceptual modelling. To this end, we design a temporal description logic, \TQL, that extends the standard ontology language \OWLQL, provides basic means for temporal conceptual modelling and ensures first-order rewritability of conjunctive queries for suitably defined data instances with validity time.
\end{abstract}

%*******************

\section{Introduction}

One of the most promising and exciting applications of description logics (DLs) is to supply ontology languages and query answering technologies for ontology-based data access (OBDA), a way of querying incomplete data sources that uses ontologies to provide additional conceptual information about the domains of interest and enrich the query vocabulary. The current W3C standard language for OBDA is \OWLQL{}, which was built on the \DL{} family of DLs~\cite{CDLLR06,dllite07}. 
To answer a conjunctive query $\q$ over an \OWLQL{} ontology $\T$ and instance data $\A$, an OBDA system first `rewrites' $\q$ and $\T$ into a new first-order query $\q'$ and then evaluates $\q'$ over $\A$ (without using the ontology). The evaluation task is performed by a conventional relational database management system.
Finding efficient and practical rewritings has been the subject of extensive research~\cite{Perez-UrbinaMH09,RosatiAKR10,KR10our,Chortaras-etal2011,2011_Gottlob,DBLP:conf/rr/KonigLMT12}. Another fundamental feature of \OWLQL, supplementing its first-order rewritability, is the ability to capture basic conceptual data modelling constructs~\cite{BeCD05,artale:et:al:er-07}.

In applications, instance data is often time-dependent: employment contracts come to an end, parliaments are elected, children are born. Temporal data can be modelled by pairs consisting of facts and their validity time; for example, $\nm{givesBirth}(\nm{diana}, \nm{william}, 1982)$. To query data with validity time, it would be useful to employ an ontology that provides a conceptual model for both static and temporal aspects of the domain of interest. Thus, when querying the fact above, one could use the knowledge that, if $x$ gives birth to $y$, then $x$ becomes a mother of $y$ from that moment on: 
\begin{equation}\label{diana}
\Diamond_{\PT}\nm{givesBirth} \sqsubseteq \nm{motherOf},
\end{equation}
where $\Diamond_{\PT}$ reads `sometime in the past.\!' 
\OWLQL{} does not support temporal conceptual modelling and, rather surprisingly, no attempt has yet been made to lift the OBDA framework to temporal ontologies and data.  

Temporal extensions of DLs have been investigated since 1993; see~\cite{GKWZ03,LuWoZa-TIME-08,AF05} for surveys and \cite{DBLP:conf/ijcai/FranconiT11,DBLP:conf/rr/Gutierrez-BasultoK12,DBLP:journals/tocl/BaaderGL12} for more recent developments. Temporalised \DL{} logics have been constructed for temporal conceptual data modelling~\cite{AKRZ:ER10}. But unfortunately, none of the existing temporal DLs supports first-order rewritability. 

The aim of this paper is to design a temporal DL that contains \OWLQL, provides basic means for temporal conceptual modelling and, at the same time, ensures first-order rewritability of conjunctive queries (for suitably defined data instances with validity time).

The temporal extension \TQL{} of \OWLQL{} we present here is interpreted over sequences 
$\I(n)$, $n \in \mathbb Z$, of standard DL structures reflecting possible evolutions of data. 
TBox axioms are interpreted globally, that is, are assumed to hold in all of the $\I(n)$, but the 
concepts and roles they contain can vary in time. ABox assertions (temporal data) are time-stamped 
unary (for concepts) and binary (for roles) predicates that hold at the specified moments of time. 
%For example, the assertion ${\it Retired}({\it bob},2012)$ and the axiom $\Diamond_P {\it Retired} \sqsubseteq {\it Pensioner}$ 
%imply that starting from 2012 Bob is a pensioner. 
Concept (role) inclusions of \TQL{} generalise \OWLQL{} inclusions by
allowing intersections of basic concepts (roles) in the left-hand
side, possibly prefixed with temporal operators $\Diamond_{\PT}$ (sometime
in the past) or $\Diamond_{\FT}$ (sometime in the future). Among other
things, one can express in \TQL{} that a concept/role name is rigid (or time-independent), persistent in the past/future or instantaneous. For example, $\Diamond_{\FT}\Diamond_{\PT}\nm{Person} \sqsubseteq \nm{Person}$ states that the concept $\nm{Person}$ is rigid, $\Diamond_{\PT}\nm{hasName}\sqsubseteq \nm{hasName}$ says
that the role $\nm{hasName}$ is persistent in the future, while $\nm{givesBirth} \sqcap \Diamond_{\PT} \nm{givesBirth} \sqsubseteq \bot$ implies that $\nm{givesBirth}$ is instantaneous.  Inclusions such as $\Diamond_{\PT} \nm{Start} \sqcap \Diamond_{\FT} \nm{End} \sqsubseteq \nm{Employed}$ represent convexity (or existential rigidity) of concepts or roles. However, in contrast to most existing temporal DLs, we cannot use temporal operators in the right-hand side of inclusions (e.g., to say that every student will eventually graduate: $\nm{Student} \sqsubseteq \Diamond_{\FT}\nm{Graduate}$).

In conjunctive queries (CQs) over \TQL{} knowledge bases, we allow time-stamped predicates together with atoms of the form $(\tau < \tau')$ or $(\tau = \tau')$, where $\tau, \tau'$ are temporal constants denoting integers or variables ranging over integers.

Our main result is that, given a \TQL{} TBox $\T$ and a CQ $\q$, one
can construct a union $\q'$ of CQs such that the answers to $\q$ over
$\T$ and any temporal ABox $\A$ can be computed by evaluating $\q'$
over $\A$ extended with the temporal precedence relation $<$ between
the moments of time in $\A$. For example, the query $\nm{motherOf}(x,y,t)$ over~\eqref{diana} can be rewritten as
\begin{equation*}
\nm{motherOf}(x,y,t) \lor \exists t' \, \big((t'<t) \land \nm{givesBirth}(x,y,t')\big).
\end{equation*}
Note that the addition of the transitive relation $<$ to the ABox is unavoidable: without it, there exists no first-order rewriting even  for the simple example above~\cite[Cor.~4.13]{Libkin}.  

From a technical viewpoint, one of the challenges we are facing is that, in contrast to known OBDA languages with CQ rewritability (including fragments of datalog$^\pm$~\cite{CaliGL12}), witnesses for existential quantifiers outside the ABox are not independent from each other but interact via the temporal precedence relation. For this reason, a reduction to known languages  appears to be impossible and a novel approach to rewriting has to be found. We also observe that straightforward temporal extensions of \TQL{} lose first-order rewritability.
For example, query answering over the ontology $\{{\it Student} \sqsubseteq \Diamond_{\FT}{\it Graduate}\}$ is shown to be non-tractable. 

\phantom{Omitted proofs can be found at}%~\cite{proofs}.}

%Finally,\nb{M: stopped here} to emphasise the main steps of the construction and the underlying principles, we do not
%consider efficiency issues in this paper. For example, replacing time points by intervals in the representation
%of data and, in particular, in $[\A]$ when we evaluate queries.

%**********************

\section{$\TQL$: a Temporal Extension of \OWLQL}

\emph{Concepts} $C$ and \emph{roles} $S$ of $\TQL$ are defined by the  grammar: 
\begin{align*}
R \ &::=\ \bot \quad \mid \quad P_i \quad \mid \quad P_i^- ,  \\
B \ &::=\  
\bot \quad \mid \quad A_i \quad \mid \quad \exists R, \\
C\ & ::=\ B \quad \mid \quad C_1 \sqcap C_2 \quad \mid \quad 
\Diamond_{\PT}C \quad \mid \quad \Diamond_{\FT}C,\\ 
%\ \mid \ \nxt_{P}C \ \mid \ \nxt_{F}C,
%
S \ & ::=\ R \quad \mid \quad S_1 \sqcap S_2 \quad \mid \quad 
\Diamond_{\PT} S \quad \mid \quad \Diamond_{\FT} S,
\end{align*}
where $A_i$ is a \emph{concept name}, $P_i$ a \emph{role name} $(i \ge 0)$, and $\Diamond_{\PT}$ and $\Diamond_{\FT}$ are temporal operators `sometime in the past' and `sometime in the future,' respectively. We call concepts and roles of the form $B$ and $R$ \emph{basic}.  
A \TQL{} \emph{TBox}, $\T$, is a finite set of \emph{concept} and \emph{role inclusions} of the form 
\begin{align*}
C \sqsubseteq B, \qquad %C \sqsubseteq \bot, \qquad 
S \sqsubseteq R, %\qquad S \sqsubseteq \bot,
\end{align*}
which are assumed to hold globally (over the whole timeline). Note that the $\Diamond_{\FT/\PT}$-free fragment of \TQL{} is an extension of the description logic \DLh~\cite{dllite-jair09} with role inclusions of the form $R_1 \sqcap \dots \sqcap R_n \sqsubseteq R$;   it properly contains \OWLQL{} (the missing role constraints can be safely added to the language). 

A \TQL{} \emph{ABox}, $\mathcal{A}$, is a (finite) set of atoms  $P_i(a,b,n)$ and $A_i(a,n)$, where $a, b$ are \emph{individual constants} and $n \in \mathbb Z$ a \emph{temporal constant}. The set of individual constants in $\A$ is denoted by $\ind(\mathcal{A})$, and the set of temporal constants by $\tc(\mathcal{A})$. A \TQL{} \emph{knowledge base} (\emph{KB}) is a pair $\mathcal{K} = (\mathcal{T}, \mathcal{A})$, where $\T$ is a TBox and $\A$ an ABox. 

A \emph{temporal interpretation}, $\mathcal{I}$, is given by the ordered set $(\mathbb Z,<)$ of \emph{time points} and standard (atemporal) interpretations $\mathcal{I}(n)= (\Delta^\I, \cdot^{\mathcal{I}(n)})$, for each $n \in \mathbb Z$. Thus, $\Delta^\I \ne \emptyset$ is the common domain of all $\I(n)$, $a_i^{\mathcal{I}(n)} \in \Delta^\I$, $A_i^{\mathcal{I}(n)} \subseteq \Delta^\I$ and $P_i^{\mathcal{I}(n)} \subseteq \Delta^\I \times \Delta^\I$. We assume that $a_i^{\mathcal{I}(n)} = a_i^{\mathcal{I}(0)}$, for all $n\in\mathbb{Z}$. To simplify presentation, we adopt the \emph{unique name assumption}, that is, $a_i^{\mathcal{I}(n)} \ne a_j^{\mathcal{I}(n)}$ for $i \ne j$ (although the obtained results hold without it). 
The role and concept constructs are interpreted in $\I$ as follows, where $n \in \mathbb Z$:
\begin{align*}
& \bot^{\mathcal{I}(n)} = \emptyset \text{\footnotesize{} (for both concepts and roles)},\\
& (P^-_i)^{\mathcal{I}(n)} = \{ (x,y) \mid (y,x) \in P_i^{\mathcal{I}(n)} \},\\
& (\exists R)^{\I(n)} = \{ x \mid  (x,y) \in R^{\I(n)}, \text{ for some } y \},\\
& (C_1 \sqcap C_2)^{\I(n)} = C_1^{\I(n)} \cap C_2^{\I(n)},\\
& (\Diamond_{\PT} C)^{\I(n)} = \{ x \mid  x \in C^{\I(m)}, \text{ for some } m < n \},\\
& (\Diamond_{\FT} C)^{\I(n)} = \{ x \mid  x \in C^{\I(m)}, \text{ for some } m > n \},\\
%
%& (\nxt_P C)^{\I(n)} = C^{\I(n-1)},\\
%
%& (\nxt_F C)^{\I(n)} = C^{\I(n+1)}.
%
& (S_1 \sqcap S_2)^{\I(n)} = S_1^{\I(n)} \cap S_2^{\I(n)},\\
& (\Diamond_{\PT} S)^{\I(n)} = \{ (x,y) \mid  (x,y) \in S^{\I(m)}, \text{ for some } m < n \},\\
& (\Diamond_{\FT} S)^{\I(n)} = \{ (x,y) \mid  (x,y) \in S^{\I(m)}, \text{ for some } m > n \}.
\end{align*}
The \emph{satisfaction relation} $\models$ is defined by taking 
\begin{align*}
& \mathcal{I} \models A_i(a,n) &&\text{iff} \qquad a^{\I(n)} \in A_i^{\I(n)},\\
& \mathcal{I} \models P_i(a,b,n) &&\text{iff} \qquad (a^{\I(n)}, b^{\I(n)}) \in P_i^{\I(n)},\\
& \mathcal{I} \models C \sqsubseteq B &&\text{iff} \qquad C^{\I(n)} \subseteq  B^{\I(n)}, \text{ for all } n \in \mathbb Z,\\
& \mathcal{I} \models S \sqsubseteq R &&\text{iff} \qquad S^{\I(n)} \subseteq R^{\I(n)}, \text{ for all } n \in \mathbb Z.
\end{align*}
If all inclusions in $\T$ and atoms in $\A$ are satisfied in $\I$, we call $\I$ a \emph{model} of $\K = (\T, \A)$ and write $\I \models \K$.
 
A \emph{conjunctive query} (\emph{CQ}) is a (two-sorted) first-order formula $\q(\vec{x}, \vec{s}) = \exists \vec{y}, \vec{t} \, \varphi(\vec{x}, \vec{y}, \vec{s}, \vec{t})$, where $\varphi(\vec{x}, \vec{y}, \vec{s}, \vec{t})$ is a conjunction of atoms of the form $A_i(\xi,\tau)$, $P_i(\xi, \zeta, \tau)$, $(\tau = \sigma)$ and $(\tau < \sigma)$, with $\xi$, $\zeta$ being \emph{individual terms}---individual constants or variables in $\vec{x}$, $\vec{y}$---and $\tau$, $\sigma$ \emph{temporal terms}---temporal constants or variables in $\vec{t}$, $\vec{s}$. In a \emph{positive existential query} (PEQ) $\q$, the formula $\varphi$ can also contain $\lor$. A \emph{union of CQs} (UCQ) is a disjunction of CQs (so every PEQ is equivalent to an exponentially larger UCQ). 

Given a KB $\K = (\T, \A)$ and a CQ $\q(\vec{x},\vec{s})$, we call tuples $\vec{a} \subseteq \ind(\A)$ and $\vec{n} \subseteq \tc(\A)$ a \emph{certain answer} to $\q(\vec{x},\vec{s})$ over $\K$ and write $\K \models \q(\vec{a},\vec{n})$, if $\I \models \q(\vec{a},\vec{n})$ for every model $\I$ of $\K$ (understood as a two-sorted first-order model). 

\begin{example}\em
Suppose Bob was a lecturer at UCL between times $n_1$ and $n_2$, after which he was appointed professor on a permanent contract. To model this situation, we use individual names, $e_1$ and $e_2$, to represent the two events of Bob's employment. The ABox will contain $n_1 < n_2$ and the atoms 
%
%\begin{multline*}
$\nm{lect}(\nm{bob}, e_1, n_1)$, $\nm{lect}(\nm{bob}, e_1, n_2)$, 
$\nm{prof}(\nm{bob}, e_2, n_2+1)$.
%\end{multline*}
%
In the TBox, we make sure that everybody is holding the corresponding post over the duration of the contract, and include other knowledge about the university life:
%
%represent standard information such as: lecturer duties include teaching and PhD supervision; professors only have to supervise PhD students, etc.:
%
\begin{align*}
&\Diamond_{\PT} \nm{lect} \sqcap \Diamond_{\FT} \nm{lect} \sqsubseteq \nm{lect}, &&
\Diamond_{\PT} \nm{prof} \sqsubseteq \nm{prof},\\
& \exists \nm{lect} \sqsubseteq \nm{Lecturer}, && \exists \nm{prof} \sqsubseteq \nm{Professor},\\
& %\nm{Lecturer} \sqsubseteq \exists \nm{supervisesPhD}, && 
\nm{Professor} \sqsubseteq \exists \nm{supervisesPhD}, &&
\nm{Professor} \sqsubseteq \nm{Staff}, \\
& \Diamond_{\PT} \nm{supervisesPhD} \sqcap \Diamond_{\FT} \nm{supervisesPhD} \sqsubseteq \nm{supervisesPhD}, \hspace*{-3cm} && \\
& && \text{etc.}
%& \nm{Lecturer} \sqsubseteq \exists \nm{teaches},\\
%%
%& \nm{Lecturer} \sqsubseteq \nm{Staff}, && \nm{Professor} \sqsubseteq \nm{Staff}.
\end{align*} 
We can now obtain staff who supervised PhDs between times $k_1$ and $k_2$ by posing the following CQ:
% $\q(x)$:
%
\begin{equation*}
\exists y,t \, \bigl( (k_1 < t < k_2) \land \nm{Staff}(x, t) \land \nm{supervisesPhD}(x,y,t) \bigr).
\end{equation*}
\end{example}

The key idea of OBDA is to reduce answering CQs over KBs to evaluating FO-queries over relational databases. To obtain such a
reduction for \TQL{} KBs, we employ a very basic type of temporal
databases. With every \TQL{} ABox $\mathcal{A}$, we associate a data
instance $[\A]$ which contains all atoms from $\A$ as well as the
atoms $(n_{1}<n_{2})$ such that $n_i\in\mathbb{Z}$ with $\min \tc(\A)
\le n_i \le \max \tc(\A)$ and $n_{1}<n_{2}$. Thus, in addition to
$\A$, we explicitly include in $[\A]$ the temporal precedence relation over the \emph{convex closure} of the time points that occur in $\A$. (Note that, in standard temporal databases, the order over timestamps is built-in.) The main result of this paper is the following:

\begin{theorem}\label{main}
Let $\q(\vec{x},\vec{s})$ be a CQ and $\T$ a \TQL{} TBox. Then one can construct a UCQ $\q'(\vec{x},\vec{s})$ such that, for any consistent KB $(\T,\A)$ such that $\A$ contains all temporal constants from $\q$, any $\vec{a} \subseteq \ind(\A)$ and $\vec{n} \subseteq \tc(\A)$, we have %of appropriate length, 
$(\T,\A) \models \q(\vec{a},\vec{n})$ iff $[\A] \models \q'(\vec{a},\vec{n})$.
\end{theorem}

Such a UCQ $\q'(\vec{x},\vec{s})$ is called a \emph{rewriting} for $\q$ and $\T$. We begin by showing how to compute rewritings for CQs over KBs with empty TBoxes.

For an ABox $\A$, we denote by $\A^{\mathbb{Z}}$ the \emph{infinite} data instance which contains the atoms in $\A$ as well as all $(n_{1}<n_{2})$ such that $n_{1},n_{2}\in \mathbb{Z}$ and $n_{1}<n_{2}$. It will be convenient to regard CQs $\q(\vec{x}, \vec{s})$ as \emph{sets} of atoms, so that we can write, e.g., $A(\xi,\tau)\in \q$. We say that $\q$ is \emph{totally ordered} if, for any temporal terms $\tau,\tau'$ in $\q$, at least one of the constraints $\tau < \tau'$, $\tau = \tau'$ or $\tau' < \tau$ is in $\q$ and the set of such constraints is consistent (in the sense that it can be satisfied in $\mathbb Z$). Clearly, every CQ is equivalent to a union of totally ordered CQs (note that the empty union is $\bot$).
\begin{lemma}\label{reduction}
For every UCQ $\q(\vec{x},\vec{s})$, one can compute a UCQ $\q'(\vec{x},\vec{s})$ such that, for any ABox $\A$ containing all temporal constants from $\q$ and any $\vec{a}\subseteq \ind(\A)$,
$\vec{n}\subseteq \tc(\A)$, we have
\begin{equation*}
\A^{\mathbb{Z}} \models \q(\vec{a},\vec{n}) \quad \mbox{ iff } \quad [\A] \models \q'(\vec{a},\vec{n}).
\end{equation*}
\end{lemma}
\begin{proof}
We assume that every CQ $\q_0$ in $\q$ is totally ordered. In each such $\q_0$, we remove a bound temporal variable $t$ together with the atoms containing $t$ if at least one of the following two conditions holds:
\begin{itemize}
\item[--] there is no temporal constant or free temporal variable $\tau$ with $(\tau<t) \in \q_{0}$, and for no temporal term $\tau'$ and atom of the form $A(\xi,\tau')$ or $P(\xi,\zeta,\tau')$ in $\q_{0}$ do we have $(\tau' <t)$ or $(\tau'=t)$ in $\q_{0}$;  

\item[--] the same as above but with $<$ replaced by $>$.
\end{itemize}
It is readily checked that the resulting UCQ is as required.
\end{proof}

\begin{example}\label{ex:flat}\em
Suppose $\T = \{\Diamond_{\FT} C \sqsubseteq A, \ \Diamond_{\PT}A\sqsubseteq B\}$ and $\q(x,s) = B(x,s)$. 
Then, for any $\A$, $a\in \ind(\A)$, $n\in \tc(\A)$, we have 
$(\T,\A) \models \q(a,n)$ iff $\A^{\mathbb{Z}} \models \q'(a,n)$,
where 
\begin{equation*}
\q'(x,s) = B(x,s) \ \ \lor \ \ \exists t \, \bigl((t < s) \land A(x,t)\bigr) \ \ 
\lor \ \ \exists t,r \, \bigl((t < s) \land (t < r) \land C(x,r)\bigr).
\end{equation*}
Note, however, that $\q'$ is \emph{not} a rewriting for $\q$ and $\T$. Take, for example, $\A=\{C(a,0)\}$. Then $(\T,\A)\models B(a,0)$ but $[\A] \not\models \q'(a,0)$. A correct rewriting is obtained by replacing the last disjunct in $\q'$ with $\exists r\, C(x,r)$; it can be computed by applying Lemma~\ref{reduction} to $\q'$ and slightly  simplifying the result.
\end{example}

In view of Lemma~\ref{reduction}, from now on we will only focus on rewritings over $\A^{\mathbb Z}$. 

The problem of finding rewritings for CQs and \TQL{} TBoxes can be reduced to the case where the TBoxes only contain inclusions of the form 
\begin{align*}
& B_{1} \sqcap B_{2} \sqsubseteq B, \qquad \Diamond_{\FT}B_{1} \sqsubseteq B_{2},\qquad 
\Diamond_{\PT}B_{1} \sqsubseteq B_{2},\\
& R_{1} \sqcap R_{2} \sqsubseteq R, \qquad \Diamond_{\FT} R_{1} \sqsubseteq R_{2},\qquad
\Diamond_{\PT}R_{1} \sqsubseteq R_{2}.
\end{align*}
We say that such TBoxes are in \emph{normal form}.

\begin{theorem}\label{normal}
For every \TQL{} TBox $\mathcal{T}$, one can construct in polynomial time a \TQL{} TBox $\mathcal{T}'$ in normal form \textup{(}possibly containing additional concept and role names\textup{)} such that $\mathcal{T}'\models \mathcal{T}$ and, for every model $\mathcal{I}$ 
of $\mathcal{T}$, there exists a model of $\mathcal{T}'$ that coincides with $\mathcal{I}$ on all concept and role names in $\mathcal{T}$.
\end{theorem}

Suppose now that we have a UCQ rewriting $\q'$ for a CQ $\q$ and the TBox $\T'$ in Theorem~\ref{normal}. We obtain a rewriting for $\q$ and $\T$ simply by removing from $\q'$ those CQs that contain symbols occurring in $\T'$ but not in $\T$. From now on, we assume that \emph{all \TQL{} TBoxes are in normal form}. The set of role names in $\T$ and with their inverses is denoted by $\roleT$, while $|\T|$ is the number of concept and role names in $\T$.

We begin the construction of rewritings by considering the case when all concept inclusions are of the form $C \sqsubseteq A_i$, so  existential quantification $\exists R$ does not occur in the right-hand side. \TQL{} TBoxes of this form will be called \emph{flat}. Note that RDFS statements can be expressed by means of flat TBoxes.

%*********************

\section{UCQ Rewriting for Flat TBoxes}\label{sec:flat}

Let $\K = (\T,\A)$ be a KB with a flat TBox $\T$ (in normal form). Our first aim is to construct a model $\C_{\K}$ of $\K$, called the \emph{canonical model}, for which the following theorem holds:
\begin{theorem}\label{can-flat}
For any consistent KB $\K = (\T,\A)$ with flat $\T$ and any CQ $\q(\vec{x}, \vec{s})$, we have $\K \models \q(\vec{a},\vec{n})$ iff $\C_\K  \models \q(\vec{a},\vec{n})$, for all tuples $\vec{a} \subseteq \ind(\A)$ and $\vec{n} \subseteq \mathbb Z$. % of appropriate length.
\end{theorem}

The construction uses a closure operator, $\op$, which applies the rules {\bf (ex)}, {\bf (c1)}--{\bf (c3)}, {\bf (r1)}--{\bf (r3)} below to a set, $\Smc$, of atoms of the form  $R(u,v,n)$, $A(u,n)$, $\exists R(u,n)$ or $(n < n')$; $\op(\Smc)$ is the result of (non-recursively) applying those rules to $\Smc$,  
\begin{equation*}
\op^{0}(\Smc) = \Smc, \qquad \op^{i+1}(\Smc)=\op(\op^{i}(\Smc)),
\qquad 
\op^{\infty}(\Smc) = \bigcup_{i \ge 0}\op^{i}(\Smc).
\end{equation*}
\begin{description}
\item[(ex)] If $R(u,v,n)\in \Smc$ then add $\exists R(u,n)$, $ \exists R^-(v,n)$ to $\Smc$;

\item[(c1)] if $(B_{1}\sqcap B_{2} \sqsubseteq B) \in \T$ and $B_{1}(u,n)$, $B_{2}(u,n) \in \Smc$, then add $B(u,n)$ to $\Smc$;

\item[(c2)] if $(\Diamond_{\PT}B\sqsubseteq B') \in \T$, $B(u,m)\in \Smc$ for some $m < n$ and $n$ occurs in $\Smc$, then add $B'(u,n)$ to $\Smc$;

\item[(c3)] if $(\Diamond_{\FT}B\sqsubseteq B') \in \T$, $B(u,m)\in \Smc$ for some $m > n$ and $n$ occurs in $\Smc$, then add $B'(u,n)$ to $\Smc$;

\item[(r1)] if $(R_{1}\sqcap R_{2} \sqsubseteq R) \in \T$ and $R_{1}(u,v,n)$, $R_{2}(u,v,n)$ are in $\Smc$, then add $R(u,v,n)$ to $\Smc$;

\item[(r2)] if $(\Diamond_{\PT}R\sqsubseteq R') \in \T$, $R(u,v,m)\in \Smc$ for some $m < n$ and $n$ occurs in $\Smc$, then add $R'(u,v,n)$ to $\Smc$;

\item[(r3)] if $(\Diamond_{\FT}R\sqsubseteq R') \in \T$, $R(u,v,m)\in \Smc$ for some $m > n$ and $n$ occurs in $\Smc$, then add $R'(u,v,n)$ to $\Smc$.
\end{description}
Note first that $\K = (\T,\A)$ is inconsistent iff $\bot \in \op^{\infty}(\A^{\mathbb Z})$. If $\K$ is consistent, we define the \emph{canonical model} $\C_\K$ of $\K$ by taking $\Delta^{\C_\K} = \ind(\A)$, $a \in A^{\C_\K(n)}$ iff $A(a,n) \in \op^{\infty}(\A^{\mathbb Z})$, and $(a,b) \in P^{\C_\K(n)}$ iff $P(a,b,n) \in \op^{\infty}(\A^{\mathbb Z})$, for $n \in \mathbb Z$. (As $\T$ is flat, atoms of the form $\exists R(u,n)$ can only be added by ({\bf ex}).) This gives us Theorem~\ref{can-flat}.
The following lemma shows that to construct $\C_\K$ we actually need only a bounded number of applications of $\op$ which does not depend on $\A$:  
\begin{lemma}\label{lem:bound}
Suppose $\T$ is a flat TBox, let $n_\T= (4 \cdot |\T|)^{4}$. Then
$\op^{\infty}(\A^{\mathbb Z}) = \op^{n_\T}(\A^{\mathbb Z})$, for any ABox $\A$. 
\end{lemma}
\begin{proof} 
It is not hard to see that $\op^{\infty}(\mathcal{S})$ can be obtained by first exhaustively applying \textbf{(r1)}--\textbf{(r3)}, then \textbf{(ex)}, and after that \textbf{(c1)}--\textbf{(c3)}. Since no recursion of \textbf{(ex)} is needed, it is sufficient to bound the recursion depth for applications of \textbf{(r1)}--\textbf{(r3)} and \textbf{(c1)}--\textbf{(c3)} separately. As both behave similarly, we focus on \textbf{(r1)}--\textbf{(r3)}. One can show that it is enough to consider ABoxes with two individuals, say $a$ and $b$, and it is not difficult to
find a bound for the recursion depth of the separated rule sets \textbf{(r1)}, \textbf{(r2)} and,
respectively, \textbf{(r1)}, \textbf{(r3)}; the interesting part of the analysis is how often
one has to alternate between applications of \textbf{(r1)}, \textbf{(r2)} and applications of 
\textbf{(r1)}, \textbf{(r3)}. The key observation here is that each alternation introduces
a fresh \emph{cross over} (i.e., a pair $(R_{1},R_{2})$ of roles such that there are $m_{1},m_{2}\in \mathbb{Z}$
with $m_{1}+1 \geq m_{2}$, $R_{1}(a,b,n)\in \mathcal{S}$
for all $n\leq m_{1}$, and $R_{2}(a,b,n)\in \mathcal{S}$ for all $n\geq m_{2}$). The number of such 
cross overs is bounded by $|\T|^{2}$, and so the number of required alternations between exhaustively 
applying \textbf{(r1)}, \textbf{(r2)}
and \textbf{(r1)}, \textbf{(r3)} is bounded by $|\T|^{2}$.
\end{proof}

We now use Lemma~\ref{lem:bound} to construct a rewriting for any flat TBox $\T$ and CQ $\q(\vec{x},\vec{s})$. For a concept $C$ and a role $S$, denote by 
$C^\sharp$ and $S^\sharp$ their standard FO-translations: e.g., $(\Diamond_{\FT} A)^\sharp (\xi,\tau) = \exists t \, ((\tau < t) \land A(\xi, t))$ and $(\exists R)^\sharp (\xi,\tau) = \exists y \, R(\xi, y, \tau)$.
Now, given a PEQ $\varphi$, we set $\varphi^{0\downarrow} = \varphi$ and define, inductively, $\varphi^{(n+1)\downarrow}$ as the result of replacing every
\begin{itemize}
\item[--] $A(\xi,\tau)$ with $A(\xi,\tau) \vee \bigvee_{(C \sqsubseteq A) \in \T}(C^{\sharp}(\xi,\tau))^{n\downarrow}$,

\item[--] $P(\xi,\zeta,\tau)$ with $P(\xi,\zeta,\tau) \vee \bigvee_{(S \sqsubseteq P) \in \T} (S^{\sharp}(\xi,\zeta,\tau))^{n\downarrow}$.
\end{itemize}
Finally, we set 
\begin{equation*}
\ext^\mathcal{T}_{\q}(\vec{x},\vec{s}) = (\q(\vec{x},\vec{s}))^{n_\T\downarrow}.
\end{equation*}
Clearly, $\ext^\mathcal{T}_{\q}(\vec{x},\vec{s})$ is a PEQ, and so can be equivalently transformed into a UCQ. Denote by $\T^\bot$ the result of replacing $\bot$ with a fresh concept name, say $F$, in all concept inclusions and with a fresh role name, say $Q$, in all role inclusions of $\T$. Clearly $(\T^\bot,\A)$ is consistent for any ABox $\A$. Let 
$\q^\bot = (\exists x,t \, F(x,t)) \lor (\exists x,y,t \, Q(x,y,t))$. 
By Theorem~\ref{can-flat} and Lemma~\ref{lem:bound}, we obtain:
\begin{theorem}\label{th:flat0}
Let $\T$ be a flat TBox and $\q(\vec{x}, \vec{s})$ a CQ. Then, for any consistent KB $(\T,\A)$, any $\vec{a}\subseteq \ind(\A)$ and $\vec{n}\subseteq \mathbb Z$,
$$
(\T,\A) \models \q(\vec{a},\vec{n}) \quad \text{iff} \quad 
\A^{\mathbb Z} \models \ext^\mathcal{T}_{\q}(\vec{a},\vec{n}).
$$
$(\T,\A)$ is inconsistent iff $(\T^\bot,\A) \models \q^\bot$. 
\end{theorem}

Thus, we obtain a rewriting for $\q$ and $\T$ using Lemma~\ref{reduction}.

%*********************

\section{Canonical Models for Arbitrary TBoxes}
Canonical models for consistent KBs $\mathcal{K} = (\mathcal{T},\mathcal{A})$ with not necessarily flat TBoxes $\T$ (in normal form) can be constructed starting from $\A^{\mathbb Z}$ and using the   rules given in the previous section together with the following one: 
\begin{description}
\item[$(\leadsto)$] if $\exists R(u,n) \in \Smc$ and $R(u,v,n) \notin \Smc$ for any $v$, then add $R(u,v,n)$ to $\Smc$, for some fresh individual name $v$; in this case we write $u \leadsto^n_R v$. 
\end{description}
Denote by $\op_1$ the closure operator under the resulting 8 rules. Again, $\K$ is inconsistent iff $\bot \in \op_1^{\infty}(\A^{\mathbb Z})$. If $\K$ is consistent, we define the \emph{canonical model} $\C_\K$ for $\K$ by the set $\op_1^{\infty}(\A^{\mathbb Z})$ in the same way as in Section~\ref{sec:flat} but taking the domain $\Delta^{\C_\K}$ to contain all the individual names in $\op_1^{\infty}(\A^{\mathbb Z})$. 
\begin{theorem}\label{thm:canonical}
For every consistent $\mathcal{K} = (\T, \A)$ and every CQ $\q(\vec{x}, \vec{s})$, we have $\mathcal{K} \models \q(\vec{a}, \vec{n})$ iff $\mathcal{C}_\mathcal{K} \models \q(\vec{a}, \vec{n})$, for any tuples $\vec{a} \subseteq \ind(\A)$ and $\vec{n} \subseteq \mathbb Z$. 
\end{theorem}
\begin{example}\label{example1}\em
Let $\K = (\T, \A)$ with $\A = \{A(a,0)\}$ and
$$
\T \ \ = \ \ \{A \sqsubseteq \exists R,  \ \Diamond_{\PT} R \sqsubseteq Q,   \ \exists Q^- \sqsubseteq \exists S, \ \Diamond_{\PT} Q \sqsubseteq P,  \ \Diamond_{\PT} S \sqsubseteq S' \}.
$$
A fragment of the model $\mathcal{C}_\K$ is shown in the picture below:
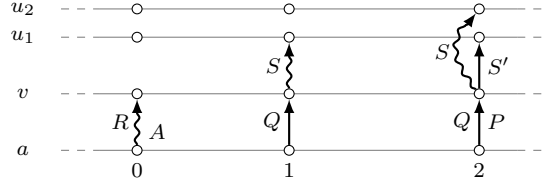
\begin{center}
\begin{tikzpicture}[>=latex, yscale=0.75, point/.style={circle,fill=white,draw=black,minimum size=1.3mm,inner sep=0pt}, wiggly/.style={thick,decorate,decoration={snake,amplitude=0.3mm,segment length=2mm,post length=1mm}},
time/.style={thick},
tw/.style={shorten <= 0.1cm, shorten >= 0.1cm,dashed}]\footnotesize
\foreach \y in {0,1,2,2.5} 
{
\draw[draw=gray] (-3.5,\y) -- ++(5.5,0);
\draw[draw=gray,dashed] (-4,\y) -- ++(0.5,0);
\draw[draw=gray,dashed] (2,\y) -- ++(0.5,0);
}
\node (a) at (-4.5,0) {$a$};
%\node (a-inf) at (-3,0) {$\cdots$};
%\node (ainf) at (3.5,0) {$\cdots$};
\node (a-1) at (-3,0) [point, label=below:{$0$}, label=above right:{$A$}]{}; % label=left:{$\exists R$}, 
\node (a0) at (-1,0) [point, label=below:{$1$}]{}; % , label=right:{$\exists Q$}
\node (a1) at (1.5,0) [point, label=below:{$2$}]{}; % , label=right:{$\exists Q, \exists P$}

\node (c) at (-4.5,1) {$v$};
%\node (c-inf) at (-3,1) {$\cdots$};
%\node (cinf) at (3.5,1) {$\cdots$};
\node (c-1) at (-3,1) [point]{}; % , label=right:{$\exists R^-$}
\node (c0) at (-1,1) [point]{}; % , label=left:{$\exists S$}, label=right:{$\exists Q^-$}
\node (c1) at (1.5,1) [point]{}; % , label=right:{$\exists Q^-,\exists P^-$}, label=left:{$\exists S, \exists S'$}

\node (d-1) at (-4.5,2) {$u_1$};
%\node (d-1-inf) at (-3,2) {$\cdots$};
%\node (d-1inf) at (3.5,2) {$\cdots$};
\node (d-1-1) at (-3,2) [point]{};
\node (d-10) at (-1,2) [point]{}; % , label=right:{$\exists S^-$}
\node (d-11) at (1.5,2) [point]{};

\node (d0) at (-4.5,2.5) {$u_2$};
%\node (d0-inf) at (-3,3) {$\cdots$};
%\node (d0inf) at (3.5,3) {$\cdots$};
\node (d0-1) at (-3,2.5) [point]{};
\node (d00) at (-1,2.5) [point]{};
\node (d01) at (1.5,2.5) [point]{}; %, label=right:{$\exists S^-$}

%\node (d1) at (-4,4) {$d_{1}$};
%\node (d1-inf) at (-3,4) {$\cdots$};
%\node (d1inf) at (3.5,4) {$\cdots$};
%\node (d1-1) at (-2,4) [point]{};
%\node (d10) at (0,4) [point]{};
%\node (d11) at (2,4) [point, label=right:{$\exists S^-$}]{};

\draw[->,wiggly] (a-1)  to node [left]{$R$} (c-1);
\draw[->,thick] (a0)  to node [left]{$Q$} (c0);
\draw[->,thick] (a1)  to node [left]{$Q$} node[right] {$P$} (c1);

%\draw[->,wiggly] (c-1)  to node [left]{$S$} (d-1-1);
\draw[->,wiggly] (c0)  to node [left]{$S$} (d-10);
\draw[->,thick] (c1)  to node [right]{$S'$} (d-11);
\draw[->,wiggly,bend left,looseness=1.2] (c1)  to node [left]{$S$} (d01);
\end{tikzpicture}
\end{center}
\end{example}

We say that the individuals $a \in \ind (\A)$ are of \emph{depth} $0$ in $\C_\K$; now, if $u$ is of depth $d$ in $\C_\K$ and $u\leadsto^n_R v$, for some $n \in \mathbb Z$ and $R$, then $v$ is of \emph{depth} $d+1$ in $\C_\K$. 
Thus, both $u_1$ and $u_2$ in Example~\ref{example1} are of depth 2 and $v$ is of depth 1. The restriction of $\C_\K$, treated as a set of atoms, to the individual names of depth $\le d$ is denoted by $\C_\K^d$. Note that this set is not necessarily closed under the rule $(\leadsto)$.

In the remainder of this section, we describe the structure of $\C_\K$, which is required for the rewriting in the next section. 
We split $\C_\K$ into two parts: one consists of the elements in $\ind(\A)$, while the other contains the fresh individuals introduced by $(\leadsto)$. As this rule always uses \emph{fresh} individuals, to understand the structure of the latter part it is enough to consider KBs of the form $\K_\T^R = (\T \cup \{A \sqsubseteq \exists R\}, \{A(a,0)\})$ with fresh $A$. 
%and describe the behaviour of concepts and roles on an individual $u$ of depth 1. 
We begin by analysing the behaviour of the atoms $R'(a,u,n)$ entailed by $R(a,u,0)$, where $a \leadsto^0_{R} u$. 
\begin{lemma}[monotonicity]\label{lem:kk}
Suppose $a \leadsto^0_{R} u$ in $\CKTR$. If either $m < n < 0$ or $0< n < m$, then $R'(a,u,n)\in \CKTR$ implies $R'(a,u,m)\in \CKTR$\textup{;} moreover, if  \mbox{$n < m = - |\roleT|$} or $|\roleT|= m  < n$, then $R'(a,u,n)\in \CKTR$ iff \mbox{$R'(a,u,m)\in \CKTR$}.
\end{lemma}

The atoms $R'(a,u,n)$ entailed by $R(a,u,0)$ in $\CKTR$ via {\bf (r1)}--{\bf (r3)}, also have an impact, via {\bf (ex)}, on the atoms of the form $B(a,n)$ and $B(u,n)$ in $\CKTR$. Thus, in Example~\ref{example1}, $R(a,v,0)$ entails $\exists Q(a,n)$, for $n > 0$.  To analyse the behaviour of such  atoms, it is helpful to assume that $\T$ is in \emph{concept normal form} (CoNF) in the following sense: for every role $R \in \roleT$, the TBox $\T$ contains 
\begin{align*}
& \exists R \sqsubseteq A^0_R, && \Diamond_{\FT}\exists R \sqsubseteq A^{-1}_{R}, && \Diamond_{\FT}A^{-m}_{R} \sqsubseteq A^{-m-1}_{R}, \\
&&& \Diamond_{\PT}\exists R \sqsubseteq A^1_{R}, && \Diamond_{\PT}A^m_{R} \sqsubseteq A^{m+1}_{R},
\end{align*}
for $0 \le m \leq |\roleT|$ and some concepts $A^i_{R}$, and 
\begin{equation*}
A^{m}_{R}\sqsubseteq \exists R', \text{ for } |m| \leq |\roleT| \text{ and } R'(a,v,m) \in \CKTR. %,\\
\end{equation*}
%%
%& A^m_{R}\sqsubseteq \exists S, \qquad \text{if } S(a,v,m) \in \CKTR.
%\end{align*}
%
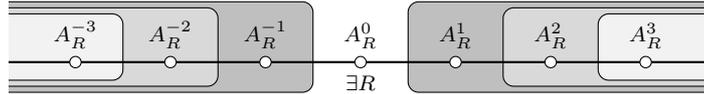
\centerline{%
\begin{tikzpicture}[>=latex, point/.style={circle,fill=white,draw=black,minimum size=1.5mm,inner sep=0pt},xscale=1.25,yscale=0.8]\footnotesize
\clip (-3.7,-0.6) rectangle ++(7.4,1.7);
\begin{scope}[rounded corners=4pt]
\draw[fill=gray!50] (-5,-0.5) rectangle ++(4.5,1.5); 
\draw[fill=gray!30] (-5,-0.4) rectangle ++(3.5,1.3); 
\draw[fill=gray!10] (-5,-0.3) rectangle ++(2.5,1.1); 
\draw[fill=gray!50] (5,-0.5) rectangle ++(-4.5,1.5); 
\draw[fill=gray!30] (5,-0.4) rectangle ++(-3.5,1.3); 
\draw[fill=gray!10] (5,-0.3) rectangle ++(-2.5,1.1); 
\end{scope}
\draw[thick] (-5,0) -- ++(10,0);
\node [point,label=below:{$\exists R$},label=above:{$A_R^0$}] at (0,0) {};
\node [point,label=above:{$A_R^{-1}$}] at (-1,0) {};
\node [point,label=above:{$A_R^{-2}$}] at (-2,0) {};
\node [point,label=above:{$A_R^{-3}$}] at (-3,0) {};
\node [point,label=above:{$A_R^{1}$}] at (1,0) {};
\node [point,label=above:{$A_R^{2}$}] at (2,0) {};
\node [point,label=above:{$A_R^{3}$}] at (3,0) {};
\end{tikzpicture}%
}\\
(In Example~\ref{example1}, $\C_\K$ will contain the atoms $A^1_R(a,n)$ and $A^2_R(a,n+1)$, for $n\ge 1$.) By Lemma~\ref{lem:kk}, if $\T$ is in CoNF, then we  can compute the atoms $B(a,n)$ and $B(u,n)$ in $\CKTR$ without using the rules {\bf (r1)}--{\bf (r3)}. Lemma~\ref{lem:kk} also implies that we can add the inclusions above (with fresh $A^i_{R}$) to $\T$ if required,  thereby obtaining a conservative extension of $\T$; so from now on we always assume \emph{$\T$ to be in CoNF}. 
These observations enable the proof of the following two lemmas. The first one characterises the atoms $B(u,n)$ in $\CKTR$:
\begin{lemma}[monotonicity]\label{l:anonym}
Suppose $a \leadsto^0_{R} u$ in $\CKTR$. If either $m < n < 0$ or $0< n < m$, then
$B(u,n)\in \CKTR$ implies $B(u,m)\in \CKTR$\textup{;} moreover, if either $n < m = -|\T|$ or $|\T|=m < n$, then $B(u,n)\in \CKTR$ iff $B(u,m)\in \CKTR$.
%
%\item[--] for any $w$ in $\CKTR$, there is $v$ of depth $\le 2|\T|$ such that $C(w,k)\in \CKTR$ implies $C(v,k)\in \CKTR$, for all $k \in \mathbb Z$.
%\end{itemize}
\end{lemma}

The second lemma characterises the ABox part of $\C_\K$ and is a straightforward generalisation of Lemma~\ref{lem:bound}:
\begin{lemma}\label{l:sructure}
For any KB $\K = (\T,\A)$ and any atom $\alpha$ of the form $A(a,n)$, $\exists R(a,n)$ or $R(a,b,n)$, where $a,b\in \ind(\A)$ and $n \in \mathbb Z$, we have $\alpha\in \C_\K$ iff $\alpha\in \op^{n_\T}(\A^{\mathbb Z})$.
\end{lemma}

An obvious extension of the rewriting of Theorem~\ref{th:flat0}  provides, for every CQ $\q(\vec{x}, \vec{s})$, a UCQ $\ext^\mathcal{T}_{\q}(\vec{x},\vec{s})$ such that for all $\vec{a}\subseteq \ind(\A)$ and $\vec{n}\subseteq \mathbb Z$ of the appropriate length,
\begin{equation}\label{groundlevel}
\C^0_{\K} \models \q(\vec{a}, \vec{n})
\quad \text{iff} \quad \A^{\mathbb{Z}} \models \ext^\mathcal{T}_{\q}(\vec{a},\vec{n}).
\end{equation}
%e
%where $(\C_{\K})_{|\A^{\mathbb{Z}}}$ is the restriction of $\C_{\K}$ to the individuals in $\A^{\mathbb{Z}}$.
%
Moreover, for a basic concept $\exists R$, we find a UCQ $\ext^\mathcal{T}_{\exists R}(\xi,\tau)$ such that, for any $a\in \ind(\A)$ and $n\in \mathbb{Z}$,
%%
%\begin{equation}\label{existsR}
$\exists R(a,n) \in \C_{\K}$  iff $\A^{\mathbb{Z}} \models 
\ext^\mathcal{T}_{\exists R}(a,n)$.

We now use the obtained results to show that one can find all answers to a CQ $\q$ over a \TQL{} KB $\K$ by only considering a fragment of $\C_\K$ whose size is polynomial in $|\T|$ and $|\q|$. This property is called the \emph{polynomial witness property}~\cite{GottlobS11}. Denote by $\mathcal{C}_\mathcal{K}^{\smash{d,\ell}}$, for $d,\ell \ge 0$, the restriction of $\mathcal{C}_\mathcal{K}^{\smash{d}}$ to the moments of time in the interval $[\min \tc(\A) - \ell, \max \tc(\A) + \ell]$.

Let $\q(\vec{x}, \vec{s})$ be a CQ. Tuples $\vec{a} \subseteq \ind(\A)$ and $\vec{n} \subseteq \tc(\A)$ give a certain answer to $\q(\vec{x}, \vec{s})$ over $\K = (\T, \A)$ iff there is a \emph{homomorphism} $h$ from $\q$ to $\C_\K$, which maps individual (temporal) terms of $\q$ to individual (respectively, temporal) terms of $\C_\K$ in such a way that the following conditions hold:
\begin{itemize}
\item[--] $h(\vec{x}) = \vec{a}$ and $h(b) = b$, for all $b\in\ind(\A)$;
\item[--] $h(\vec{s}) = \vec{n}$ and $h(m) = m$, for all $m\in\tc(\A)$; %$b\in\ind(\A)$ and ;
% individual and temporal constants $b$ and $m$; 
\item[--] $h(\q) \subseteq \C_{\K}$,
%\item[--] if $(\tau_1 < \tau_2) \in \q$ then $h(\tau_1) < h(\tau_2)$;
%
%\item[--] if $(\tau_1 = \tau_2) \in \q$ then $h(\tau_1) = h(\tau_2)$;  
%
%\item[--] if $A(\xi,\tau) \in \q$ then $A(h(\xi), h(\tau)) \in \C_\K$;
%
%\item[--] if $P(\xi, \zeta,\tau) \in \q$ then $P(h(\xi), h(\zeta), h(\tau)) \in \C_\K$.
\end{itemize}
where $h(\q)$ denotes the set of atoms obtained by replacing every term in $\q$ with its $h$-image, e.g., $P(\xi, \zeta,\tau)$ with $P(h(\xi), h(\zeta), h(\tau))$, $(\tau_1 < \tau_2)$ with $h(\tau_1) < h(\tau_2)$, etc.

Now, using the monotonicity lemmas for the temporal dimension and the fact that atoms of depth $> |\roleT|$ in the canonical models duplicate atoms of smaller depth, we obtain
\begin{theorem}\label{thm:canonical-bounded}
There are polynomials $f_1$ and $f_2$ such that, for any consistent  \TQL{} KB $\mathcal{K} = (\mathcal{T}, \mathcal{A})$, any CQ $\q(\vec{x}, \vec{s})$ and any $\vec{a} \subseteq \ind(\A)$ and $\vec{n} \subseteq \tc(\A)$, we have $\mathcal{K} \models \q(\vec{a}, \vec{n})$ iff there is a homomorphism $h\colon \q \to \C_\K$ such that  $h(\q) \subseteq \C^{\smash{d,\ell}}_\K$, where $d = f_1(|\mathcal{T}|,|\q|)$ and $\ell = f_2(|\mathcal{T}|,|\q|)$.
\end{theorem}

We are now in a position to define a rewriting for any given CQ and \TQL{} TBox.

%**********************

\section{UCQ Rewriting}

Suppose $\q(\vec{x}, \vec{s})$ is a CQ and $\T$ a \TQL{} TBox (in CoNF). Without loss of generality we assume $\q$ to be totally ordered. 
By a \emph{sub-query} of $\q$ we understand any subset $\q' \subseteq \q$ containing all temporal constraints 
\mbox{$(\tau < \tau')$} and $(\tau = \tau')$ that occur in $\q$.
In the rewriting for $\q$ and $\T$ given below, we consider all possible splittings of $\q$ into two sub-queries (sharing the same temporal terms). One is to be mapped to the ABox part of the canonical model $\C_{(\T,\A)}$, and so we can rewrite it using~\eqref{groundlevel}. The other sub-query is to be mapped to the non-ABox part of $\C_{(\T,\A)}$ and requires a different rewriting. 

For every $R \in \roleT$, we construct the set $\CKTR^{\smash{d,\ell}}$, where $d$ and $\ell$
% = f_1(|\mathcal{T}|,|\q|)$ and $\ell = f_2(|\mathcal{T}|,|\q|)$ 
are provided by Theorem~\ref{thm:canonical-bounded}. 
Let $h$ be a map from a sub-query $\q_h \subseteq \q$ to $\CKTR^{\smash{d,\ell}}$ such that $h(\q_h) \subseteq \CKTR^{\smash{d,\ell}}$.
Denote by $\mathcal{X}_h$ the set of individual terms $\xi$ in $\q_h$ with $h(\xi) = a$, and let $\mathcal{Y}_h$ be the remaining set of individual terms in $\q_h$.
We call $h$ a \emph{witness for} $R$ if  
\begin{itemize}\itemsep=0pt
\item[--] $\mathcal{X}_h$ contains at most one individual constant;

\item[--] every term in $\mathcal{Y}_h$ is a quantified variable in $\q$;

\item[--] $\q_h$ contains all atoms in $\q$ with a variable from $\mathcal{Y}_h$.
\end{itemize}

Let $h$ be a witness for $R$. Denote by $\leadsto$ the union of all $\leadsto_{R'}^n$ in $\CKTR^{\smash{d,\ell}}$. Clearly, $\leadsto$ is a tree order on the individuals in $\CKTR^{\smash{d,\ell}}$, with root $a$. Let $T_h$ be its minimal sub-tree containing $a$ and the $h$-images of all the individual terms in $\q_h$. For each $v \in T_h\setminus\{a\}$, we take the (unique) moment $\g(v)$ with $u \leadsto_R^{\g(v)} v$, for some $u$ and $R$, and set $\g(a) = 0$.
For $A(y,\tau) \in \q_h$, we say that $h(y)$ \emph{realises} $A(y,\tau)$. For any $P(\xi,\xi',\tau) \in \q_h$, there are $u,u'\in T_h$ with $u \leadsto u'$ and $\{u,u'\} = \{h(\xi),h(\xi')\}$; we say that $u'$ \emph{realises} $P(\xi,\xi',\tau)$. 
Let $\vec{r}$ be a list of fresh temporal variables $r_u$, for $u\in T_h\setminus\{a\}$. Consider the following formula, whose free variables are $r_a$ and the temporal variables of $\q_h$:
\begin{equation*}
\mathsf{t}_h = \exists \vec{r}\, \bigl( \bigwedge_{u\leadsto v} \delta^{\g(v) - \g(u)}(r_{u},r_v) \land{} \hspace*{-2em} \bigwedge_{u \text{ realises } \alpha(\vec{\xi},\tau)}\hspace*{-2em} \delta^{h(\tau) - \g(u)}(r_u,\tau) \bigr),
\end{equation*}
% 
%
%
%Then for every such atom $\alpha(\vec{y},\tau) \in \q_h$ with history $(n_1,\dots,n_d,n)$, we define
%%
%\begin{equation*}
%H_{\alpha(\vec{y},\tau)}(r_0,r) =  \exists r_1,\dots,r_d\, \bigl(\bigwedge_{i=1}^{d} \delta^{n_i}(r_{i-1},r_i) \land \delta^n (r_d,r)\bigr),
%\end{equation*}
%
where the formulas 
%$\delta^0(t,s)$ is $(t=s)$, 
$\delta^n(t,s)$ say that $t$ is at least $n$ moments before $s$: that is, $\delta^0(t,s)$ is $(t=s)$ and $\delta^n(t,s)$ is
\begin{align*}
& \exists s_1,\dots,s_{n-1} (t < s_1 < \dots < s_{n-1} < s), &  \text{ if $n >0$},\\
& \exists s_1,\dots,s_{|n|-1} (t > s_1 > \dots > s_{|n|-1} > s), & \text{ if $n <0$}.
\end{align*}
%
%
%
%Let $d$ be the maximal depth of $h(y)$ in $\C^{k,\ell}_{\K_{\T,R}}$, for $y \in \mathcal{Y}_h$, and suppose also that $(a,n_0) \leadsto_{R_0} (u_0,n_0) \leadsto_{R_1}  \dots \leadsto_{R_{d-1}} (u_{d-1},n_{d-1})$, where $n_0 =0$ and $R_0 = R$. Let $m_i = n_i - n_{i-1}$, for $1 \le i \le d-1$. 
%If $\mathcal{X}_h \ne \emptyset$, we pick some $\xi_h \in \mathcal{X}_h$, which is  an individual constant if it exists, and 
Take a fresh variable $x_h$ and associate with $h$ the formula 
\begin{align*}%\label{tw-formula}
\tw_h =~ & \exists r_a \exists x_h \, \bigl[ \ext^\mathcal{T}_{\exists R}(x_h, r_a) \ \land  \bigwedge_{h(\xi) = a} \hspace*{-2mm} (\xi = x_h)  \ \ \land \ \ \mathsf{t}_h\bigr].
%
%& \exists r_0,\dots,r_{d-1} \, ( \delta^0(t,r_0) \land  \bigwedge_{i=1}^{d-1} \delta^{m_i}(r_{i-1},r_i) \land {}\\
%
%& \bigwedge_{\begin{subarray}{c} h(\xi) = u_i \\ C(\xi,\tau) \in \q_h\end{subarray}} \delta^{h(\tau) - n_i}(r_i,\tau) \land \bigwedge_{\begin{subarray}{c} h(\xi) = u_i \\ P(\xi,\zeta,\tau) \in \q_h\end{subarray}} \delta^{h(\tau) - n_i}(r_i,\tau))\big]
\end{align*}
%
%If $\mathcal{X}_h = \emptyset$, we take $\xi_h$ to be a fresh individual variable $x$ and add $\exists x$ to the prefix of $\tw_h$.
%
To give the intuition behind $\tw_h$, suppose that $\C_{(\T,\A)} \models^g \tw_h$, for some assignment $g$. Then $g$ maps all terms in $\mathcal{X}_h$ to $g(x_h) \in \ind(\A)$ such that $\exists R (g(x_h), g(r_a)) \in \C_{(\T,\A)}$, so $(g(x_h), g(r_a))$ is the root of a substructure of $\C_{(\T,\A)}$ isomorphic to $\CKTR$ in which the variables from $\mathcal{Y}_h$ can be mapped according to $h$. For temporal terms, the formula $\mathsf{t}_h$ cannot specify the values prescribed by $h$: without $\neg$ in UCQs, we can only say that $\tau$ is at least (not exactly) $n$ moments before $\tau'$. However, by Lemmas~\ref{lem:kk} and~\ref{l:anonym}, this is still enough to ensure that $g$ and $h$ give a homomorphism from $\q_h$ to $\C_{(\T,\A)}$.
\begin{example}\em
Let $\T$ be the same as in Example~\ref{example1} and let 
\begin{align*}
\q(x,t) = \exists y,z,t'\,\big( (t < t')\land  Q(x,y,t) \land S'(y,z,t')\big).
\end{align*}
The map $h = \{ x\mapsto a, y \mapsto v, z\mapsto u_1, t \mapsto 1, t' \mapsto 2\}$ is a witness for $R$, with $\q_h = \q$ and $\tw_h$ is the following formula
\begin{equation*}
\exists r_a\exists x_h \, \bigl(\ext^\mathcal{T}_{\exists R}(x_h,r_a) \land (x_h = x) \land
\exists r_v\exists r_{u_1}\bigl(\delta^0(r_a,r_v) \land \delta^1(r_v,r_{u_1}) \land \delta^1(r_v,t) \land \delta^1(r_{u_1},t')\bigr)\bigr).
\end{equation*}
\end{example}

We can now define a rewriting for $\q(\vec{x}, \vec{s})$ and $\mathcal{T}$. Let $\mathfrak{T}$ be the set of all witnesses for $\q$ and $\mathcal{T}$. We call a subset $\mathfrak S \subseteq \mathfrak{T}$  \emph{consistent} if $(\mathcal{X}_{h_1} \cup \mathcal{Y}_{h_1}) \cap (\mathcal{X}_{h_1}\cup\mathcal{Y}_{h_2}) \subseteq \mathcal{X}_{h_1}\cap \mathcal{X}_{h_2}$, for any distinct $h_1, h_2 \in \mathfrak S$.   Assuming that $\vec{y}$ contains all the quantified variables in $\q$ and $\q \setminus \mathfrak S$ is the sub-query of $\q$ obtained by removing the atoms in $\q_h$, $h \in \mathfrak S$, other than $(\tau < \tau')$ and $(\tau = \tau')$, we set:
\begin{align*}%\label{long-rew}
\q^*(\vec{x}, \vec{s}) ~=~ \exists \vec{y}  \bigvee_{\begin{subarray}{c}
\mathfrak S \subseteq \mathfrak{T}\\ \mathfrak S \text{ consistent}
\end{subarray}} \hspace*{-1mm}
 \big(
\bigwedge_{h \in \mathfrak S} \tw_h  \ \land \ \ext^\T_{\q \setminus \mathfrak S}\big).
%\hspace*{-3mm} \bigwedge_{
%\begin{subarray}{c}
%\alpha(\vec{z}) \in \q\\
%\forall h \in \mathfrak S \, \alpha(\vec{z}) \notin \q_h
%\end{subarray}
%} \hspace*{-4mm}
%\mathsf{ext}^\mathcal{T}_{\alpha} (\vec{z})  \land \q_\mathsf{t} \Big).
\end{align*}
\begin{theorem}\label{thm:rewriting}
Let $\T$ be a \TQL{} TBox in CoNF and $\q(\vec{x}, \vec{s})$ a totally ordered  CQ. Then, for any consistent KB $(\T,\A)$ and any tuples $\vec{a} \subseteq \ind(\A)$ and $\vec{n} \subseteq \mathbb Z$, 
$$
(\T, \A) \models \q(\vec{a},\vec{n}) \quad \text{iff} \quad \A^{\mathbb Z} \models \q^*(\vec{a},\vec{n}).
$$
$(\T,\A)$ is inconsistent iff $(\T^\bot,\A) \models \q^\bot$.
\end{theorem}

%We obtain a rewriting for $\q$ and $\T$ using Lemma~\ref{reduction}.

Theorem~\ref{main} now follows by Lemma~\ref{reduction}.

%****************

\section{Non-Rewritability}\label{sec:non-rewritability}

In this section, we show that the language \TQL{} is nearly optimal as far as rewritability of CQs and ontologies is concerned.

We note first, that the syntax of \TQL{} allows concept inclusions and role inclusions; `mixed' axioms such as the datalog rule $A(x,t) \wedge R(x,y,t) \rightarrow B(x,t)$ are not expressible. The reason is that mixed rules often lead to non-rewritability, as is well known from the DL $\mathcal{EL}$. For example, there does not exist an FO-query $\q(x,t)$ such that $(\T,\A) \models A(a,n)$ iff
$\A^{\mathbb Z}\models \q(a,n)$ for $\T = \{A(y,t) \wedge R(x,y,t) \rightarrow A(x,t)\}$ since such a
query has to express that at time-point $t$ there is an $R$-path from $x$ to some $y$ with $A(y,t)$.

Second, it would seem to be natural to extend \TQL{} with the temporal next/previous-time operators as concept or role constructs.
However, again this would lead to non-rewritability: any FO-rewriting for $A(x,t)$ and $\{\nxt_{\PT}A \sqsubseteq B, \, \nxt_{\PT}B \sqsubseteq A\}$ has to express that there exists $n\geq 0$ such that $A(x,t-2n)$ or $B(x,t-(2n+1))$, which is impossible~\cite{Libkin}.

Another natural extension would be inclusions of the form $A \sqsubseteq \Diamond_{\FT}B$.
(Note that inclusions of the form $A \sqsubseteq \exists R.B$ are expressible in \OWLQL.) But again such an extension would ruin rewritability. The reason is that temporal precedence $<$ is a total order, and so one can construct an ABox $\A$
and a UCQ $\q(x) = \q_{1}\vee \q_{2}$ such that $(\T,\A)\models \q(a)$ but $(\T,\A)\not\models \q_{i}(a)$, $i=1,2$, for $\T = \{A\sqsubseteq \Diamond_{\FT}B\}$. Indeed, we take $\A = \{ A(a,0),C(a,1)\}$ and
\begin{align*}
& \q_{1}(x)= \exists t \, ( C(x,t) \wedge B(x,t)),\\
& \q_{2}(x)= \exists t,t' \, ( (t<t') \wedge C(x,t) \wedge B(x,t')).
\end{align*}
In fact, by reduction of 2+2-SAT~\cite{Schaerf-93}, we prove the following:
\begin{theorem}\label{thm:conphard}
Answering CQs over the TBox $\{A \sqsubseteq \Diamond_{\FT}B\}$ is \coNP-hard for data complexity.
\end{theorem}

%****************

\section{Related Work}

 The Semantic Web community has developed a variety
 of extensions of RDF/S and OWL with validity time 
 \cite{DBLP:journals/ws/Motik12,DBLP:conf/www/PuglieseUS08,DBLP:journals/tkde/GutierrezHV07}.
 The focus of this line of research is on representing
 and querying time stamped RDF triples or OWL axioms. 
 In contrast, in our language only instance data are time stamped , while the ontology formulates time independent constraints 
 that describe how the extensions of concepts and roles 
 can change over time. In the temporal DL literature, a similar distinction
 has been discussed as the difference between temporalised
 axioms and temporalised concepts/roles; the expressive power of
 the respective languages is incomparable~\cite{GKWZ03,DBLP:journals/tocl/BaaderGL12}.
 
In Theorem~\ref{th:flat0}, we show rewritability using boundedness
of recursion. This connection between first-order definability
and boundedness is well known from the datalog and logic literature
where boundedness has been investigated extensively 
\cite{DBLP:conf/lics/GaifmanMSV87,DBLP:journals/ijfcs/Meyden00,DBLP:conf/icalp/KreutzerOS07}.
 \citeauthor{DBLP:journals/corr/abs-0902-1179} [2009] investigate
 boundedness for datalog programs on linear orders; the results 
 are different from ours since the linear
 order is the only predicate symbol of the datalog programs considered
 and no further restrictions (comparable to ours) are imposed.

%*********************

\section{Conclusion}

In this paper, we have proved UCQ rewritability for conjunctive queries and \TQL{} ontologies over data instances with validity time. Our focus was solely on the existence of rewritings, and we did not consider efficiency issues such as finding shortest rewritings, using temporal intervals in the data representation or mappings between temporal databases and ontologies. We only note here that these issues are of practical importance and will be addressed in future work. It would also be of interest to investigate the possibilities to increase the expressive power of both ontology and query language. For example, we believe that the extension of \TQL{} with the next/previous time operators, which can only occur in TBox axioms not involved in cycles, will still enjoy rewritability. We can also increase the expressivity of conjunctive queries by allowing the arithmetic operations $+$ and $\times$ over temporal terms, which would make the CQ $A(x,t)$ and the TBox $\{\nxt_{\PT}A \sqsubseteq B, \, \nxt_{\PT}B \sqsubseteq A\}$ rewritable in the extended language.

%Mappings from temporal databases in an OBDA system?

%Note\nb{do we need this} that, in temporal databases, we can timestamp relations with temporal intervals; so we can, for example, write $\nm{lect}(\nm{bob}, e_1, [n_1,n_2])$. 

%*********************

\newpage

\bibliographystyle{named}
\bibliography{local,DL-Lite-bib}

\begin{thebibliography}{}

\bibitem[\protect\citeauthoryear{Artale and Franconi}{2005}]{AF05}
A.~Artale and E.~Franconi.
\newblock Temporal description logics.
\newblock In {\em Handbook of Temporal Reasoning in Artificial Intelligence},
  Foundations of Artificial Intelligence, pages 375--388. Elsevier, 2005.

\bibitem[\protect\citeauthoryear{Artale \bgroup \em et al.\egroup
  }{2007}]{artale:et:al:er-07}
A.~Artale, D.~Calvanese, R.~Kontchakov, V.~Ryzhikov, and M.~Zakharyaschev.
\newblock Reasoning over extended {ER} models.
\newblock In {\em Proc.\ of ER-2007}, volume 4801 of {\em LNCS}, pages
  277--292. Springer, 2007.

\bibitem[\protect\citeauthoryear{Artale \bgroup \em et al.\egroup
  }{2009}]{dllite-jair09}
A.~Artale, D.~Calvanese, R.~Kontchakov, and M.~Zakharyaschev.
\newblock The {DL-Lite} family and relations.
\newblock {\em J.\ of Artifical Intelligence Research}, 36:1--69, 2009.

\bibitem[\protect\citeauthoryear{Artale \bgroup \em et al.\egroup
  }{2010}]{AKRZ:ER10}
A.~Artale, R.~Kontchakov, V.~Ryzhikov, and M.~Zakharyaschev.
\newblock Complexity of reasoning over temporal data models.
\newblock In {\em Proc.\ of ER-2010}, volume 6412 of {\em LNCS}, pages
  174--187. Springer, 2010.

\bibitem[\protect\citeauthoryear{Baader \bgroup \em et al.\egroup
  }{2012}]{DBLP:journals/tocl/BaaderGL12}
F.~Baader, S.~Ghilardi, and C.~Lutz.
\newblock {LTL} over description logic axioms.
\newblock {\em ACM Trans. Comput. Log.}, 13(3):21, 2012.

\bibitem[\protect\citeauthoryear{Berardi \bgroup \em et al.\egroup
  }{2005}]{BeCD05}
D.~Berardi, D.~Calvanese, and G.~De~Giacomo.
\newblock Reasoning on {UML} class diagrams.
\newblock {\em Artificial Intelligence}, 168(1--2):70--118, 2005.

\bibitem[\protect\citeauthoryear{Cal\`{\i} \bgroup \em et al.\egroup
  }{2012}]{CaliGL12}
A.~Cal\`{\i}, G.~Gottlob, and T.~Lukasiewicz.
\newblock A general datalog-based framework for tractable query answering over
  ontologies.
\newblock {\em J. Web Sem.}, 14:57--83, 2012.

\bibitem[\protect\citeauthoryear{Calvanese \bgroup \em et al.\egroup
  }{2006}]{CDLLR06}
D.~Calvanese, G.~De~Giacomo, D.~Lembo, M.~Lenzerini, and R.~Rosati.
\newblock Data complexity of query answering in description logics.
\newblock In {\em Proc.\ of KR 2006}, pages 260--270, 2006.

\bibitem[\protect\citeauthoryear{Calvanese \bgroup \em et al.\egroup
  }{2007}]{dllite07}
D.~Calvanese, G.~{De Giacomo}, D.~Lembo, M.~Lenzerini, and R.~Rosati.
\newblock Tractable reasoning and efficient query answering in description
  logics: The {DL-Lite} family.
\newblock {\em J.\ of Automated Reasoning}, 39(3):385--429, 2007.

\bibitem[\protect\citeauthoryear{Chortaras \bgroup \em et al.\egroup
  }{2011}]{Chortaras-etal2011}
A.~Chortaras, D.~Trivela, and G.~Stamou.
\newblock Goal-oriented query rewriting for {OWL 2 QL}.
\newblock In {\em Proc. of DL 2011}, volume 745 of {\em CEUR-WS}, 2011.

\bibitem[\protect\citeauthoryear{Franconi and
  Toman}{2011}]{DBLP:conf/ijcai/FranconiT11}
E.~Franconi and D.~Toman.
\newblock Fixpoints in temporal description logics.
\newblock In {\em Proc.\ of IJCAI 2011}, pages 875--880, 2011.

\bibitem[\protect\citeauthoryear{Gabbay \bgroup \em et al.\egroup
  }{2003}]{GKWZ03}
D.~Gabbay, A.~Kurucz, F.~Wolter, and M.~Zakharyaschev.
\newblock {\em Many-dimensional modal logics: theory and applications}.
\newblock Studies in Logic. Elsevier, 2003.

\bibitem[\protect\citeauthoryear{Gaifman \bgroup \em et al.\egroup
  }{1987}]{DBLP:conf/lics/GaifmanMSV87}
H.~Gaifman, H.~G. Mairson, Y.~Sagiv, and M.~Y. Vardi.
\newblock Undecidable optimization problems for database logic programs.
\newblock In {\em Proc.\ of LICS 87}, pages 106--115, 1987.

\bibitem[\protect\citeauthoryear{Gottlob and Schwentick}{2011}]{GottlobS11}
G.~Gottlob and T.~Schwentick.
\newblock Rewriting ontological queries into small nonrecursive datalog
  programs.
\newblock In {\em Proc.\ of DL 2011}, volume 745 of {\em CEUR-WS}, 2011.

\bibitem[\protect\citeauthoryear{Gottlob \bgroup \em et al.\egroup
  }{2011}]{2011_Gottlob}
G.~Gottlob, G.~Orsi, and A.~Pieris.
\newblock Ontological queries: Rewriting and optimization.
\newblock In {\em Proc.\ of ICDE 2011}, pages 2--13. IEEE Computer Society,
  2011.

\bibitem[\protect\citeauthoryear{Guti{\'e}rrez-Basulto and
  Klarman}{2012}]{DBLP:conf/rr/Gutierrez-BasultoK12}
V.~Guti{\'e}rrez-Basulto and S.~Klarman.
\newblock Towards a unifying approach to representing and querying temporal
  data in description logics.
\newblock In {\em Proc.\ of RR 2012}, volume 7497 of {\em LNCS}, pages 90--105.
  Springer, 2012.

\bibitem[\protect\citeauthoryear{Gutierrez \bgroup \em et al.\egroup
  }{2007}]{DBLP:journals/tkde/GutierrezHV07}
C.~Gutierrez, C.~A. Hurtado, and A.~A. Vaisman.
\newblock Introducing time into {RDF}.
\newblock {\em IEEE Trans. Knowl. Data Eng.}, 19(2):207--218, 2007.

\bibitem[\protect\citeauthoryear{K{\"o}nig \bgroup \em et al.\egroup
  }{2012}]{DBLP:conf/rr/KonigLMT12}
M.~K{\"o}nig, M.~Lecl{\`e}re, M.-L. Mugnier, and M.~Thomazo.
\newblock A sound and complete backward chaining algorithm for existential
  rules.
\newblock In {\em Proc.\ of RR 2012}, volume 7497 of {\em LNCS}, pages
  122--138. Springer, 2012.

\bibitem[\protect\citeauthoryear{Kontchakov \bgroup \em et al.\egroup
  }{2010}]{KR10our}
R.~Kontchakov, C.~Lutz, D.~Toman, F.~Wolter, and M.~Zakharyaschev.
\newblock The combined approach to query answering in {DL-Lite}.
\newblock In {\em Proc.\ of KR 2010}. AAAI Press, 2010.

\bibitem[\protect\citeauthoryear{Kreutzer \bgroup \em et al.\egroup
  }{2007}]{DBLP:conf/icalp/KreutzerOS07}
S.~Kreutzer, M.~Otto, and N.~Schweikardt.
\newblock Boundedness of monadic {FO} over acyclic structures.
\newblock In {\em Proc.\ of ICALP 2007}, pages 571--582, 2007.

\bibitem[\protect\citeauthoryear{Libkin}{2004}]{Libkin}
L.~Libkin.
\newblock {\em Elements Of Finite Model Theory}.
\newblock Springer, 2004.

\bibitem[\protect\citeauthoryear{Lutz \bgroup \em et al.\egroup
  }{2008}]{LuWoZa-TIME-08}
C.~Lutz, F.~Wolter, and M.~Zakharyaschev.
\newblock Temporal description logics: A survey.
\newblock In {\em Proc.\ of TIME 2008}, pages 3--14. IEEE Computer Society,
  2008.

\bibitem[\protect\citeauthoryear{Motik}{2012}]{DBLP:journals/ws/Motik12}
B.~Motik.
\newblock Representing and querying validity time in {RDF} and {OWL}: A
  logic-based approach.
\newblock {\em J. Web Sem.}, 12:3--21, 2012.

\bibitem[\protect\citeauthoryear{P{\'e}rez-Urbina \bgroup \em et al.\egroup
  }{2009}]{Perez-UrbinaMH09}
H.~P{\'e}rez-Urbina, B.~Motik, and I.~Horrocks.
\newblock A comparison of query rewriting techniques for {DL-Lite}.
\newblock In {\em Proc.\ of DL 2009}, volume 477 of {\em CEUR-WS}, 2009.

\bibitem[\protect\citeauthoryear{Pugliese \bgroup \em et al.\egroup
  }{2008}]{DBLP:conf/www/PuglieseUS08}
A.~Pugliese, O.~Udrea, and V.~S. Subrahmanian.
\newblock Scaling {RDF} with time.
\newblock In {\em Proc.\ of WWW 2008}, pages 605--614, 2008.

\bibitem[\protect\citeauthoryear{Rosati and Almatelli}{2010}]{RosatiAKR10}
R.~Rosati and A.~Almatelli.
\newblock Improving query answering over {DL-Lite} ontologies.
\newblock In {\em Proc.\ of KR 2010}. AAAI Press, 2010.

\bibitem[\protect\citeauthoryear{Schaerf}{1993}]{Schaerf-93}
A.~Schaerf.
\newblock On the complexity of the instance checking problem in concept
  languages with existential quantification.
\newblock {\em J.\ of Intel.\ Inf.\ Systems}, 2:265--278, 1993.

\bibitem[\protect\citeauthoryear{van~der
  Meyden}{2000}]{DBLP:journals/ijfcs/Meyden00}
R.~van~der Meyden.
\newblock Predicate boundedness of linear monadic datalog is in {PSPACE}.
\newblock {\em Int. J. Found. Comput. Sci.}, 11(4):591--612, 2000.

\end{thebibliography}

%\end{document}

%****************

\cleardoublepage

\appendix

\section{Proofs}

We first give a detailed definition of the standard translation $C^{\sharp}$ and $S^{\sharp}$ of
\TQL{} concepts $C$ and roles $S$ into two-sorted first-order logic.
The definitions are by induction as follows. For concepts:
\begin{eqnarray*}
A^{\sharp}(\xi,\tau) &  = & A(\xi,\tau),\\
\bot^{\sharp}(\xi,\tau) & = & \bot,\\
(\exists R)^{\sharp}(\xi,\tau) & = & (\exists y \;R^{\sharp}(\xi,y,\tau)),\\
%(\exists P^{-})^{\sharp}(\xi,\tau) & = & (\exists y \;P(y,\xi,\tau))\\
(C_{1}\sqcap C_{2})^{\sharp}(\xi,\tau) & = & C_{1}^{\sharp}(\xi,\tau)\wedge C_{2}^{\sharp}(\xi,\tau),\\ 
(\Diamond_{\PT}C)^{\sharp}(\xi,\tau) & = & \exists t \;((t<\tau) \wedge C^{\sharp}(\xi,t)),\\
(\Diamond_{\FT}C)^{\sharp}(\xi,\tau) & = & \exists t \;((t>\tau) \wedge C^{\sharp}(\xi,t)).
\end{eqnarray*}
For roles:
\begin{eqnarray*}
P^{\sharp}(\xi,\zeta,\tau) &  = & P(\xi,\zeta,\tau),\\
(P^{-})^{\sharp}(\xi,\zeta,\tau) & = & P(\zeta,\xi,\tau),\\
\bot^{\sharp}(\xi,\zeta,\tau) & = & \bot,\\
(S_{1}\sqcap S_{2})^{\sharp}(\xi,\zeta,\tau) & = & S_{1}^{\sharp}(\xi,\zeta,\tau)\wedge S_{2}^{\sharp}(\xi,\zeta,\tau),\\ 
(\Diamond_{\PT}S)^{\sharp}(\xi,\zeta,\tau) & = & \exists t \;((t<\tau) \wedge S^{\sharp}(\xi,\zeta,t)),\\
(\Diamond_{\FT}S)^{\sharp}(\xi,\zeta,\tau) & = & \exists t \;((t>\tau) \wedge S^{\sharp}(\xi,\zeta,t)).
\end{eqnarray*}

\medskip
 
\noindent
{\bf Lemma~\ref{lem:bound} }{\it
Suppose $\T$ is a flat TBox, let $n_\T= (4 \cdot |\T|)^{4}$. Then
$\op^{\infty}(\A^{\mathbb Z}) = \op^{n_\T}(\A^{\mathbb Z})$, for any ABox $\A$.}

\smallskip

\begin{proof}
We start with an observation that to compute $\op^{\infty}(\mathcal{S})$ it is sufficient
to first compute the closure under the rules \textbf{(r1)}--\textbf{(r3)} for role inclusions, 
then apply the rule \textbf{(ex)},
and then apply the rules \textbf{(c1)}--\textbf{(c3)} for concept inclusions. Formally, for a set $R$ of rules, 
let $\op_{R}(\mathcal{S})$ denote the result of applying the rules in $R$
(non-recursively!) to $\mathcal{S}$. Let ${\sf role}=\{\textbf{(r1)},\textbf{(r2)}, \textbf{(r3)}\}$
and ${\sf concept}= \{\textbf{(c1)},\textbf{(c2)},\textbf{(c3)}\}$. Then

\medskip

\noindent
\textbf{Fact 1.} $\op^{\infty}(\mathcal{S})= \op_{{\sf concept}}^{\infty}(\op_{\{\textbf{(ex)}\}}(\op^{\infty}_{{\sf role}}(\mathcal{S})))$, 
for all $\mathcal{S}$.

\medskip

The proof is straightforward: since none of the rules \textbf{(c1)}--\textbf{(c3)} or \textbf{(ex)} introduces a new role assertion, i.e.,
an assertion of the form $R(u,v,n)$,
no new applications of rules \textbf{(r1)}--\textbf{(r3)} become possible after applying rules  
\textbf{(c1)}--\textbf{(c3)} and \textbf{(ex)}; and no new applications of \textbf{(ex)} becomes possible after applying rules \textbf{(c1)}--\textbf{(c3)}.

\medskip

%Let $m_{r} = |\roleT|$.
We first consider the closure under the rules for role inclusions and show
$\op^{\infty}_{{\sf role}}(\A^{\mathbb{Z}}) = \op^{k_{r}}_{{\sf role}}(\A^{\mathbb{Z}})$
for $k_{r}=4|\roleT|^{4}$. We start with the observation that it is sufficient to show
this for ABoxes having at most two individuals because role assertions for individuals $u,v$
do not interact with role assertions for individuals $u',v'$ if $\{u,v\}\not=\{u',v'\}$.
Formally, for any $u,v$ ($u=v$ is not excluded), let $\A_{u,v}$ consist of all assertions $R(u,v,n)$
in $\A$. Then 

\medskip

\noindent
\textbf{Fact 2.} $\op^{k}_{{\sf role}}(\A^{\mathbb{Z}}) = \bigcup_{u,v\in {\sf Ind}(\A)}\op^{k}_{{\sf role}}(\A_{u,v}^{\mathbb{Z}})$,
for all $k\geq 0$.

\medskip

\noindent
Now let $\A$ be an ABox with individuals $u,v$.  
Observe that the rule \textbf{(r1)} is \emph{local} in the sense that the
addition of a role assertion at time point $n$ depends only on role assertions
that hold already at time point $n$. It follows that $\op_{\{\textbf{(r1)}\}}^{\infty}(\mathcal{S})= 
\op_{\{\textbf{(r1)}\}}^{|\roleT|}(\mathcal{S})$.
We now analyse the two operators obtained by adding to \textbf{(r1)} either the rule \textbf{(r2)} or the
rule \textbf{(r3)}. Let $P=\{\textbf{(r1)},\textbf{(r2)}\}$ and $F=\{\textbf{(r1)},\textbf{(r3)}\}$. For the rules in $P$ the addition of role
assertions at time point $n$ only depends on the time points $m\leq n$ and, similarly, for
$F$ the addition of role assertions at time point $n$ only depends on time points $m\geq n$.
It is now easy to see that in each case, one has to alternate between applications of
local rules and the rule $\textbf{(r2)}$ (respectively $\textbf{(r3)}$) at most $|\roleT|$ times. Thus
%\nb{$\op_P$ and $\op_F$ each require $|\roleT|\cdot(|\roleT| + 1)\leq 2|\roleT|^2$ rule applications}

\medskip

\noindent
\textbf{Fact 3.} $\op_{P}^{\infty}(\A^{\mathbb{Z}})= \op_{P}^{|\roleT|}(\A^{\mathbb{Z}})$ and 
$\op_{F}^{\infty}(\A^{\mathbb{Z}})= \op_{F}^{|\roleT|}(\A^{\mathbb{Z}})$. %, for $k_{0}= 2 m_{r}^{2}$.

\medskip

By Fact~3, to obtain a $k_{r}$ such that 
$\op_{{\sf role}}^{\infty}(\A^{\mathbb{Z}})= \op_{{\sf role}}^{k_{r}}(\A^{\mathbb{Z}})$
it is sufficient to determine an upper bound for the number of alternations between
$\op_{P}^{\infty}$ and $\op_{F}^{\infty}$ that are required to compute $\op_{{\sf role}}^{\infty}$:

\medskip

\noindent
\textbf{Fact 4.} $\op_{{\sf role}}^{\infty}(\mathcal{S}) = (\op_{P}^{\infty}\circ\op_{F}^{\infty})^{|\roleT|^2}(\mathcal{S})$.
%where $m= m_{r}^{2}$.

\medskip

To prove Fact~4 we introduce the notion of a \emph{cross over}. Assume $u,v$ are the individuals of $\mathcal{S}$.
Let $R_{1}$ and $R_{2}$ be roles. We say that $(R_{1},R_{2})$ are a cross over in $\mathcal{S}$ 
if there are $m_{1},m_{2}$ with $m_{1}+1 \geq m_{2}$ such that $R_{1}(u,v,n)\in \mathcal{S}$
for all $n\leq m_{1}$ and $R_{2}(u,v,n)\in \mathcal{S}$ for all $n\geq m_{2}$.
 
\medskip

\noindent 
\textit{Claim 1.} Let $\mathcal{S}= \op_{P}^{\infty}(\mathcal{S})$, $\mathcal{S}_{1}= \op_{F}^{\infty}(\mathcal{S})$ and 
$\mathcal{S}_{2}= \op_{P}^{\infty}(\mathcal{S}_{1})$. If $\mathcal{S}_{2}\supsetneq \mathcal{S}_{1}$
then there exists a cross over $(R_{1},R_{2})$ in $\mathcal{S}_{2}$ 
which is not a cross over in $\mathcal{S}$. 

\medskip
\noindent\textit{Proof of Claim~1.}
Since $\mathcal{S}_{1}$ is closed under \textbf{(r1)}, there exist 
$\Diamond_{\PT}R\sqsubseteq R'$ in $\T$ and $n_{1}$
such that 
\begin{itemize}
\item[--] $R(u,v,n_{1}) \in \mathcal{S}_{1}$ and there is $n>n_{1}$ with $R'(u,v,n_{1})\notin \mathcal{S}_{1}$;
\item[--] and $R'(u,v,n)\in \mathcal{S}_{2}$, for all $n>n_{1}$.
\end{itemize}
It follows that 
$R(u,v,n_{1})\not\in \mathcal{S}$;
for otherwise $R'(u,v,m)\in \mathcal{S}_{1}$ for all $n>n_{1}$. From $R(u,v,n_{1}) \in \mathcal{S}_{1}$
and, since $\mathcal{S}$ is closed under \textbf{(r1)}, there exist 
$\Diamond_{\FT}S \sqsubseteq S'$ in  $\T$ and $n_{2}>n_{1}$ such that
\begin{itemize}
\item[--] $S(u,v,n_{2})\in \mathcal{S}$ and $S'(u,v,n_{1})\notin \mathcal{S}$;
\item[--] $S'(u,v,n)\in \mathcal{S}_{1}$, for all $n<n_{2}$.
\end{itemize}
Let $m_{1}=n_{2}-1$, $m_{2}=n_{1}+1$. %and $(R_{1},R_{2})= (S',R')$.
Then $(S',R')$ is a cross over in $\mathcal{S}_{2}$ with witness times points $m_{1},m_{2}$
which is not a cross over in $\mathcal{S}$. This finishes the proof of Claim~1.

\medskip

Clearly the number of cross overs is bounded by $|\roleT|^{2}$ and so we have proved Fact~4.
We obtain from Fact~3 and Fact~4 that $\op^{\infty}_{{\sf role}}(\A^{\mathbb{Z}}) = 
\op^{k_{r}}_{{\sf role}}(\A^{\mathbb{Z}})$ for $k_{r}= 4 |\roleT|^{4}$.
%\nb{$2\cdot |\roleT|\cdot(|\roleT + 1|) \cdot |\roleT|^2 \leq 4\cdot |\roleT|^4$}

One can show in almost exactly the same way that $\op^{\infty}_{{\sf concept}}(\A^{\mathbb{Z}}) = 
\op^{k_{c}}_{{\sf concept}}(\A^{\mathbb{Z}})$ for $k_{c}=4 m_{c}^{4}$, where $m_{c}$ is the
number of concept names in $\T$. Since $4 |\roleT|^{4}+4 m_{c}^{4} \leq (4|\T|)^{4}$, this
finishes the proof of Lemma~\ref{lem:bound}.
\end{proof}

\medskip
\noindent
{\bf Lemma~\ref{lem:kk}}
\emph{Let $a \leadsto_{R}^{0} u$ in $\CKTR$.  
Then the following hold, for all basic roles $R'$\textup{:} 
\begin{itemize}
\item[1.] if $m < n < 0$ or $0< n < m$, then $R'(a,u,n)\in \CKTR$ implies $R'(a,u,m)\in \CKTR$\textup{;}
\item[2.] if $n < m = - |\roleT|$ or $|\roleT| = m < n$, then $R'(a,u,n)\in \CKTR$ iff 
$R'(a,u,m)\in \CKTR$.
\end{itemize}
}

\begin{proof}
We start with Item~1.
Assume $0<n<m$ (the case $m<n<0$ is similar and left to the reader). 
The proof is by induction over rule applications. Namely, we show 

\medskip

\noindent
{\bf Claim~1.} Let $0<n<m$.  If $R'(a,u,n)\in \op_{1}^{k}(\A^{\mathbb{Z}})$, then
$R'(a,u,m)\in \op_{1}^{k}(\A^{\mathbb{Z}})$, for all $k\geq 0$.

\medskip
 
For $\A^{\mathbb{Z}}$ itself the claim is trivial. Now assume it holds for $\op_{1}^{k}(\A^{\mathbb{Z}})$. 
Applications of \textbf{(ex)} and \textbf{(c1)}--\textbf{(c3)} do not influence the role
assertions for $(a,u)$, so we do not have to consider them. 
Applications of \textbf{(r1)}--\textbf{(r3)} clearly preserve the property stated in 
Claim~1.

\medskip

Now consider Item~2. Assume $|\roleT| = m < n$ (the case $- |\roleT| =  m > n$ is similar and left
to the reader). The proof is by induction over rule applications. In detail, we show the following 

\medskip

\noindent
{\bf Claim~2.} For all $k\geq 0$ and $\ell>1$, if there exists $R'$ with $R'(a,u,\ell)\in 
\op_{1}^{k}(\A^{\mathbb{Z}})$ and $R'(a,u,\ell-1)\not\in \op_{1}^{k}(\A^{\mathbb{Z}})$,
then $|\{R'' \mid R''(a,u,\ell)\in \op_{1}^{k}(\A^{\mathbb{Z}})\}|\geq \ell$.

\smallskip

The proof of Claim~2 is by induction over $k$ and left to the reader.
Now Item~2 follows directly with Point~1.
\end{proof}

We now analyse in more detail why one can without loss of generality assume TBoxes to be in CoNF.
Let $\T$ be a TBox in normal form. Add to $\T$ the inclusions
\begin{align}\label{eq:CoNF1}
& \exists R \sqsubseteq A^0_R, && \Diamond_{F}\exists R \sqsubseteq A^{-1}_{R}, && \Diamond_{F}A^{-m}_{R} \sqsubseteq A^{-m-1}_{R}, \\
\label{eq:CoNF2}
&&& \Diamond_{P}\exists R \sqsubseteq A^1_{R}, && \Diamond_{P}A^m_{R} \sqsubseteq A^{m+1}_{R},
\end{align}
for all $R\in \roleT$ and $0\leq m\leq |\roleT|$, where the $A_{R}^{i}$ are fresh concept
names; and add  
\begin{equation}\label{eq:CoNF:roles}
A^{m}_{R}\sqsubseteq \exists R', \text{ for } |m| \leq |\roleT| \text{ and } R'(a,v,m) \in \CKTR. %,\\
\end{equation}
Denote the resulting TBox by $\T'$. We first show that $\T'$ is a conservative extension of $\T$ in the
following sense:
\begin{lemma}\label{cons}
For every model $\I$ of $\T$ there exists a model $\I'$ of $\T'$ such that $\Delta^{\I}=\Delta^{\I'}$ and 
such that $\I'$ coincides with $\I$ for the interpretation of symbols from $\T$. 
\end{lemma}
\begin{proof}
Assume $\I$ is given.
We define $\I'$ as follows:
\begin{align*}
A^{0}_{R})^{\I(n)} &= (\exists R)^{\I(n)},\\
(A^{m}_{R})^{\I(n)}  & = (\Diamond_{\PT}^{m}\exists R)^{\I(n)} \quad\text{ and }\quad (A^{-m}_{R})^{\I(n)}  = (\Diamond_{\FT}^{m}\exists R)^{\I(n)},\qquad \text{ for } 0< m\leq |\roleT|.
\end{align*}
We have to show that $\I'\models\T'$. 
It is readily seen that $\I'$ satisfies the inclusions~\eqref{eq:CoNF1} and~\eqref{eq:CoNF2},
for all $R\in \roleT$ and $0\leq m\leq |\roleT|$.
The interesting part are the fresh inclusions $A^{m}_{R}\sqsubseteq \exists R'$ for $R'(a,v,m) \in \CKTR$.
Let $A^{m}_{R}\sqsubseteq \exists R'$ be such a fresh inclusion. 
Consider $m\geq 0$ and let $d\in (A^{m}_{R})^{\I(n)}$. Then $d\in (\Diamond_{P}^{m}\exists R)^{\I(n)}$.
Moreover, $R'(a,v,m) \in \CKTR$ implies, by Lemma~\ref{lem:kk}, $\T \models \Diamond_{P}^{m}R \sqsubseteq R'$.
Since $\I$ is a model of $\T$, we obtain $d\in (\exists R')^{\I(n)}$.
\end{proof}

It also follows from Lemma~\ref{cons} that 
the set of role assertions in $\CKTR$ coincides with the set of role assertions in $\C_{\K_{\T'\!\!,R}}$ and so, $\T'$ is in CoNF.

Now observe that if a TBox $\T$ is in CoNF, then one can construct $\C_{\T,\A}$ by 
\begin{itemize}
\item[--] first applying the rules {\bf (r1)}--{\bf (r3)} exhaustively to ABox individuals, 
\item[--] then applying the rules {\bf (ex)}, $(\leadsto)$ and {\bf (c1)}--{\bf (c3)} exhaustively, 
\item[--] and finally applying again {\bf (r1)}--{\bf (r3)}.  
\end{itemize}
This follows from the second part of Lemma~\ref{lem:kk} (according to which the role assertions are stable at distances
larger than $|\roleT|$ in $\CKTR$) and the inclusions~\eqref{eq:CoNF:roles}.

\medskip

\noindent
{\bf Lemma~\ref{l:anonym}}
\emph{Let $\T$ be in CoNF and $a \leadsto_{R}^{0} u$ in $\CKTR$. Then the following hold, for all basic concepts $B$\textup{:}
\begin{itemize}
\item[--] if $m < n < 0$ or $0< n < m$, then $B(u,n)\in \CKTR$ implies $B(u,m)\in \CKTR$\textup{;} 
\item[--] if $n<m =-|\T|$ or $|\T| = m < n$, then  $B(u,n)\in \CKTR$ iff 
$B(u,m)\in \CKTR$.
\end{itemize}
}
\medskip

\begin{proof}
The proof is similar to the proof of Lemma~\ref{lem:kk} and omitted.
\end{proof}

To generalize the rewriting from flat TBoxes to arbitrary TBoxes, we admit PEQs
that can contain atoms of the form $\hat{\exists R}(\xi,\tau)$, where $\exists R$ is a basic concept.
Moreover, we modify the standard translation $C^{\sharp}$ to a translation $C^{\hat{\sharp}}$ that
regards basic concepts $\exists R$ as atoms by setting 
$(\exists R)^{\hat{\sharp}}(\xi,\tau)=\hat{\exists R}(\xi,\tau)$.
We now show how the definition of $\varphi^{n\downarrow}$ is modified.
Given a generalized PEQ $\varphi$, we set $\varphi^{0\downarrow}=\varphi$ and define, inductively, 
$\varphi^{(n+1)\downarrow}$ as the result of replacing 
\begin{itemize}
\item[--] every $A(\xi,\tau)$ in $\varphi$ with 
$\displaystyle A(\xi,\tau) \vee \bigvee_{C \sqsubseteq A \in \T}(C^{\hat{\sharp}}(\xi,\tau))^{n\downarrow}$, 
\item[--] every $P(\xi,\zeta,\tau)$ in $\varphi$ with $\displaystyle P(\xi,\zeta,\tau) \vee \bigvee_{S \sqsubseteq P\in \T}
(S^{\sharp}(\xi,\zeta,\tau))^{n\downarrow}$,
%\nb{do we assume that only role names occur on the right-hand side of role inclusions?}
%
\item[--] every $\hat{\exists R}(\xi,\tau)$ in $\varphi$ with  $\displaystyle 
%(\exists y\; R(\xi,y,\tau))^{(n+1)\downarrow}
\exists y\,R^{\sharp}(\xi,\zeta,\tau) \vee \bigvee_{S \sqsubseteq R\in \T}
\exists y\,(S^{\sharp}(\xi,\zeta,\tau))^{n\downarrow}
\vee  \bigvee_{C \sqsubseteq \exists R \in \T}(C^{\hat{\sharp}}(\xi,\tau))^{n\downarrow}$.
\end{itemize}
For a query $\q(\vec{x},\vec{s})$ we now define as ${\sf ext}_{\q}^{\T}(\vec{x},\vec{s})$
the result of replacing every atom $A(\xi,\tau)$ by $(A(\xi,\tau))^{n_{\T}\downarrow}$
and every atom $P(\xi,\zeta,\tau)$ by $(P(\xi,\zeta,\tau))^{n_{\T}\downarrow}$ and replacing
in the resulting PEQ the atoms $\hat{(\exists R)}(\xi,\tau)$ by $(\exists y\, R^{\sharp}(\xi,y,\tau))$.
One can readily show that 
\begin{equation}\tag{\ref{groundlevel}}
\C_{\K}^{0} \models \q(\vec{a},\vec{n}) \quad \mbox{ iff } \quad \A^{\mathbb{Z}} \models 
{\sf ext}_{\q}^{\T}(\vec{a},\vec{n}).
\end{equation}
We also set  ${\sf ext}_{\exists R}^{\T}(\xi,\tau)= (\hat{\exists R}(\xi,\tau))^{n_{\T}\downarrow}$.
Then one can show
\begin{equation}\label{existsR}
\exists R(a,n)\in \C_{\K} \quad \mbox{ iff } \quad \A^{\mathbb{Z}} \models 
{\sf ext}_{\exists R}^{\T}(a,n).
\end{equation}

\medskip

To prove Theorem~\ref{thm:canonical-bounded}, we require some preparation.
Firstly, variations of the monotonicity Lemmas~\ref{lem:kk} and \ref{l:anonym} can be proved 
for ABox individuals in arbitrary ABoxes as well.
\begin{lemma}\label{lem:abox}
For any $\K=(\T,\A)$ with $\T$ in CoNF
%and all $m,n$ 
the following hold\textup{:}  
\begin{itemize}
\item[--] if $m < n < \min\tc(\A)$ or $\max\tc(\A) < n < m$, then 
\begin{align*}
R(a,b,n)\in \C_{\K}  \ \ \ &\text{ implies } \ \ \ R(a,b,m)\in \C_{\K},\quad\text{ for all basic roles } R \text{ and any } a,b\in \ind(\A),\\
B(a,n)\in \C_{\K} \ \ \ &\text{ implies } \ \ \ B(a,m)\in \C_{\K}, \quad\text{ for all basic concepts } B  \text{ and any } a\in\ind(\A);
\end{align*}
\item[--] if $n < m = \min\tc(\A)- |\roleT|$ or $\max\tc(\A)+|\roleT| = m < n$, 
then 
\begin{align*}
R(a,b,n)\in \C_{\K} \ \ \ &\text{ iff } \ \ \ R(a,b,m)\in \C_{\K}quad\text{ for all basic roles } R \text{ and any } a,b\in \ind(\A),\\
B(a,n)\in \C_{\K} \ \ \ &\text{ iff } \ \ \ B(a,m)\in \C_{\K}, \quad\text{ for all basic concepts } B  \text{ and any } a\in\ind(\A).
\end{align*}
\end{itemize}
\end{lemma}
\begin{proof}
The proof is similar to the proof of Lemma~\ref{lem:kk} and omitted.
\end{proof}

In what follows it will often be useful to work with types. Given $(u,n)$
we denote by $\type(u,n)$ the set of basic concepts $B$ with $B(u,n)\in \C_{\K}$ and 
given $(u,n),(v,n)$ we denote by $\type(u,v,n)$ the set of basic roles $R$ with $R(u,v,n)\in \C_{\K}$,
where the knowledge base $\K$ will always be clear from the context.

Secondly, it will be useful to introduce a notation system for the individuals $u$ and
pairs $(u,n)$ in $\C_{\K}$. In detail, we identify any $u$ in $\C_{\K}$ with a vector
\begin{equation*}
(a,n_{0},R_{0},n_{1},\ldots, n_k,R_{k})
\end{equation*}
which is defined inductively as follows. 
If $u$ has depth $0$, then $u=a\in \ind(\A)$
and we denote $u$ by the singleton vector $(a)$. If $u$ has depth $k+1$ then there is
a unique $v$ of depth $k$ and $v \leadsto_{R}^{n} u$. So, if
$$
v = (a,n_{0},R_{0},\ldots,n_{k},R_{k})
$$
then we set
$$
u  = (a,n_{0},R_{0},\ldots,n_{k},R_{k},n_{k+1},R_{k+1}),
$$
where $R_{k+1}=R$ and $n_{k+1}= n - (n_{0}+\cdots +n_{k})$. 
Moreover, if $u= (a,n_{0},R_{0},\ldots,n_{k},R_{k})$, then the pair $(u,n)$ is identified with
\begin{equation*}
(a,n_{0},R_{0},n_{1},\ldots, n_k,R_{k},n_{k+1}),
\end{equation*}
where $n_{k+1}= n-(n_{0}+\cdots +n_{k})$. Observe that we can recover the time point
$n$ of any pair $(u,n)$ %representedas $(a,n_{0},R_{0},n_{1},\ldots, R_{k},n_{k+1})$ 
as 
$n= n_{0}+\cdots + n_{k+1}$.

Observe that not every vector of this from % $(a,n_{0},R_{0},n_{1},\ldots, R_{k})$ 
is identical
to some individual $u$ in $\C_{\K}$. It is easy to see which ones are, however:
$(a,n_{0},R_{0},n_{1},\ldots, R_{k})$ is identical to some $u$ in $\C_{\K}$ iff,
inductively,
\begin{equation*}
a\in \ind(\A) \quad\text{ and }\quad \exists R_{i}\in \type(a,n_{0},R_{0},\ldots,n_{i}), \text{ for all } 0\leq i \leq k. 
\end{equation*}

\begin{lemma}\label{gettingclose}
For any $u$ in $\CKTR$, there is $v$ of depth $\le 2|\roleT|$ such that 
$B(u,k)\in \CKTR$ implies $B(v,k)\in \CKTR$, for all $k \in \mathbb Z$ and all
basic concepts $B$.
\end{lemma}
\begin{proof}
Let  $u = (a,n_{0},R_{0},n_{1},\ldots, R_{m})$.
%Our aim is to construct $v$ of depth $\le 2|\roleT|$ with $B(w,k)\in \CKTR$ 
%implies $B(v,k)\in \CKTR$, for all $k \in \mathbb Z$ and all basic concepts $B$.
Assume the depth of $u$ exceeds $2|\roleT|$; that is, $m\geq 2|\roleT|$.
We construct a $v$ of depth smaller than $u$ such that $B(u,k)\in \CKTR$ 
implies $B(v,k)\in \CKTR$, for all $k \in \mathbb Z$ and all basic concepts $B$.
The claim then follows by applying the construction again until the depth of
the resulting individual does not exceed $2|\roleT|$. Suppose there is $i < j$ with $R_i = R_j$. 
Let $\delta_{ij} = n_{i+1} + \dots + n_j$.

\medskip

\noindent \textit{Case 1.} If $\delta_{ij} = 0$  
then we remove the sequence $(n_{i+1},R_{i+1},\ldots,n_{j},R_{j})$ from $u$  and set 
\begin{equation*}
v  =   (a,n_{0},R_{0},\ldots, n_{i},R_{i},\hspace*{1em}n_{j+1},\ldots, R_{m}).
\end{equation*}
Clearly, $v$ belongs to $\CKTR$ and $\type(u,k)= \type(v,k)$  for all $k\in \mathbb{Z}$.

\medskip

\noindent
\textit{Case 2.} If $\delta_{ij} > 0$ and $n_{j+1} > 0$ 
then we  remove $(n_{i+1},R_{i+1},\ldots,n_{j},R_{j})$,
replace $n_{j+1}$ by $n_{j+1} + \delta_{ij}$ and set
\begin{equation*}
v   =  (a,n_{0},R_{0},\ldots, n_{i},R_{i},\hspace*{1em} n_{j+1} + \delta_{ij}, R_{j+1},n_{j+2},\ldots, R_{m}).
\end{equation*}
Note that
we have $0<n_{j+1}< n_{j+1} + \delta_{ij}$ and so, by Lemma~\ref{l:anonym},
\begin{equation*}
\type(a,n_{0},R_{0},\ldots, n_j,R_j,n_{j+1})
\subseteq 
\type(a,n_{0},R_{0},\ldots, n_{i},R_{i},n_{j+1}+\delta_{ij}).
\end{equation*}
Hence, $v$ belongs to $\CKTR$ and 
$\type(u,k)\subseteq \type(v,k)$ for all $k\in \mathbb{Z}$.

\medskip

\noindent
\textit{Case 3.} If $\delta_{ij} <0$ and  $n_{j+1}< 0$ then this case is dual to Case 2.

\medskip

\noindent
\textit{Case 4.} If $\delta_{ij} > 0$ and
 $n_{j+1} < 0$ and there exists $j'>j$ such that $n_{j'+1}>0$
then we remove  $(n_{i+1},R_{i+1},\ldots,n_{j},R_{j})$, replace
$n_{j'+1}$ by $n_{j'+1} + \delta_{ij}$ and set
\begin{equation*}
v  =   (a,n_{0},R_{0},\ldots, n_{i},R_{i},\hspace*{1em} n_{j+1},R_{j+1},\ldots,R_{j'},n_{j'+1}+\delta_{ij},R_{j'+1},n_{j'+2},\ldots, R_{m}).
\end{equation*}
We have $0 < n_{j'+1}<\delta_{ij}+n_{j'+1}$ and
so, by Lemma~\ref{l:anonym}, 
\begin{equation*}
\type(a,n_{0},R_{0},\ldots,n_{j'},R_{j'},n_{j'+1})
\subseteq \type(a,n_{0},R_{0},\ldots, n_{i},R_{i},n_{j+1},R_{j+1},\ldots,R_{j'},n_{j'+1}+\delta_{ij})
\end{equation*}
It follows that $v$ belongs to $\CKTR$ with 
$\type(u,k)\subseteq \type(v,k)$ for all $k\in \mathbb{Z}$.

\medskip

\noindent
\textit{Case 5.}  If $\delta_{ij} < 0$, $n_{j+1}>0$ and there exists $j'>j$ such that $n_{j'+1}<0$ then this case is dual to Case~4.\\\mbox{}
\end{proof}

\medskip
\noindent
{\bf Theorem~\ref{thm:canonical-bounded}}
\emph{
There are polynomials $f_1$ and $f_2$ such that, for any consistent  \TQL{} KB $\mathcal{K} = (\mathcal{T}, \mathcal{A})$, any CQ $\q(\vec{x}, \vec{s})$ and any $\vec{a} \subseteq \ind(\A)$ and $\vec{n} \subseteq \tc(\A)$, we have $\mathcal{K} \models \q(\vec{a}, \vec{n})$ iff there is a homomorphism $h\colon \q \to \C_\K$ such that  $h(\q) \subseteq \C^{\smash{d,\ell}}_\K$, where $d = f_1(|\mathcal{T}|,|\q|)$ and $\ell = f_2(|\mathcal{T}|,|\q|)$.}

\medskip

\noindent
\begin{proof}
Assume that a homomorphism $h\colon \q \rightarrow \C_{\K}$ is given.
A sub-query of $\q$ is a non-empty subset $\q' \subseteq \q$
containing all temporal atoms in $\q$.
We assume that $\q$ is totally ordered. First we consider the parameter $d$ for
the depth required to find a match for $\q$ in $\C_{\K}^{d}$.
We set %$f_{1}(|\T|,|\q|) 
$d= 2|\roleT| + |\q|$.

A sub-query $\q'$ of $\q$ is \emph{connected} if for any two individual terms $\xi_{1},\xi_{2}$ in $\q'$
there are $\tau_{1},\ldots,\tau_{n}$ and $\zeta_{0},\ldots,\zeta_{n}$ such that there are roles $R_{1},
\ldots,R_{n}$ with  $R_{1}(\zeta_{0},\zeta_{1},\tau_{1}),\ldots,R_{n}(\zeta_{n-1},\zeta_{n},\tau_{n})\in \q'$
such that $\zeta_{0}=\xi_{1}$ and $\zeta_{n}=\xi_{2}$.
% and
%\begin{itemize}
%\item either $\taux=x_{0}$ or there exists $A(x,t_{1})\in \q'$;
%\item either $y = x_{n}$ or there exists $A(y,t_{n})\in \q'$.
%\end{itemize}
We consider the maximal connected components $\q_{1},\ldots,\q_{n}$ of $\q$
and construct the new homomorphism $h'$ as follows:
\begin{itemize}
\item[--] $h'(\tau)=h(\tau)$ for all temporal terms $\tau$;
\item[--] For all $\q_{i}$ such that $h(\xi)\in \ind(\A)$ for some 
$\xi$ in $\q_{i}$ we have $h(y) \in \C_{\K}^{|\q_{i}|}$ for all $y$ in $\q_{i}$. Thus,
for individual terms $\xi$ in such a $\q_{i}$ we set $h'(\xi)=h(\xi)$.
\item[--] Let $\q_{i}$ be such that $h(\xi)\not\in \C_{\K}^{2|\roleT|+|\q|}$ for some $\xi$ in $\q_{i}$.
Take $\xi$ from $\q_{i}$ such that $h(\xi)$ is of minimal depth in $\C_{\K}$, which, by our assumption, exceeds $2|\roleT|$. By Lemma~\ref{gettingclose}, we find $v$ of depth $\leq 2|\roleT|$ such that
$\type(h(\xi),m) \subseteq \type(v,m)$ for all $m\in \mathbb{Z}$. Assume
\begin{equation*}
v=(a,n_{0},R_0,\ldots,n_{k},R_{k})\quad \text{ and }\quad
h(\xi) = (b,m_{0},R_0',\ldots,m_{r},R_{r}').
\end{equation*}
Let $\xi'$ be an individual variable in $\q_{i}$. 
Then $h(\xi')$ is the concatenation of $h(\xi)$ and some vector 
\begin{equation*}
(m_{r+1},R_{r+1}',m_{r+2},\ldots, m_{s},R_{s}').
\end{equation*}
Denote $m= (m_{0}+\cdots + m_{r})-(n_{0}+\cdots+ n_{k})$
and define 
\begin{equation*}
h'(\xi') \ \ =  \ \ (a,n_{0},\ldots,n_{k},R_{k}, \ \ 
m_{r+1}+m,R_{r+1}',m_{r+2},\ldots,m_{s},R_{s}').
\end{equation*}
\end{itemize}
The resulting $h'$ is the required homomorphism.
%after this has been done for all $\q_{i}$ is as required.

\medskip

Now we consider the bound $\ell$ for the ``width'' of the match for $\q$. 
Assume that $h\colon \q \rightarrow \C_{\K}^{d}$ for $d= 2|\roleT| + |\q|$.
We set $\ell = |\T|\cdot |\q|\cdot d$
%f_{2}(|\T|,|\q|)$, where 
%
%$$
%f_{2}(|\T|,|\q|)= ,
%$$
%
and transform $h$ into a homomorphism $h'\colon \q \rightarrow \C_{\K}^{\smash{d,\ell}}$. It should be clear that the required polynomials are $f_1(|\T|,|\q|) = 4|\T| + |\q|$ and $f_2(|\T|,|\q|) = |\T| \cdot |\q| \cdot f_1(|\T|,|\q|)$, respectively. 

%For a vector $g=(a,n_{0},R_{0},\ldots,R_{k-1},n_{k})$ representing a pair $(v,n)$ we define
%the \emph{temporal extension} $g^{t}$ of $g$ as the set of all time points
%$$
%n_{0},n_{0}+n_{1},\ldots,n_{0}+\cdots+n_{k}
%$$
%Similarly, 

We define the \emph{temporal extension} $u^t$ of an individual $u$ with representation
$(a,n_{0},R_{0},\ldots,R_{k-1},n_{k},R_{k})$ as the set $\{\overline{n}_{0},\overline{n}_{1},\ldots,\overline{n}_{k}\}$, where
\begin{equation*}
\overline{n}_{k} = n_{0}+\cdots+n_{k}.
\end{equation*}
% 
%The temporal extension of a set $\mathcal{X}$ of such vectors is defined as
%$\mathcal{X}^{t}=\bigcup_{u\in \mathcal{X}}u^{t}$. 
%
By $\mathcal{H}^t_{h}$ we denote the set of all $h(\tau)$ and all $(h(\xi))^t$, with $h(\xi)$ are represented as vectors as introduced above. Let $\mathcal{M}=\{\min\tc(\A)-|\T|,\max\tc(\A)+|\T|\}$.
The homomorphism $h'$ we are going to construct will have the following property:
\begin{itemize}
\item $|m_1-m_2|\leq |\T|$, for any two $m_{1},m_{2} \in \mathcal{H}_{h'}^{t}\cup \mathcal{M}$ such that there is no 
$m\in  \mathcal{H}_{h'}^{t}$ between $m_1,m_2$ and such that
$m_1,m_2\leq\min\tc(\A)-|\T|$ or $m_1,m_2\geq\max\tc(\A)+|\T|$.
\end{itemize}
Assume such an $h'$ has been constructed. The cardinality of $\mathcal{H}_{h'}^{t}$ is
bounded by $|\q|\cdot d$ (since $h'$ is into $\C_{\K}^d$).
Hence $h'$ is into $\C_{\K}^{\smash{d,\ell}}$, %for $\ell = |\T|\times |\q|\times (2|\roleT| + |\q|)$,
as required.

For the construction of $h'$, let $m_1,m_2 \in \mathcal{H}_{h}^{t}\cup \mathcal{M}$ be such that
there is no $m\in  \mathcal{H}_{h}^{t}$ between $m_1$ and $m_2$ and such that
$m_1<m_2\leq \min\tc(\A)-|\T|$ or $\max\tc(\A)+|\T|\leq m_1<m_2$. Assume $|m_2-m_1|>|\T|$
and that, without loss of generality, $\max\tc(\A)+|\T|\leq m_1<m_2$.
We define $h'$ as follows. Let $m= (m_2-m_1)-|\T|$.
Then, for all temporal terms $\tau$ we set 
\begin{equation*}
h'(\tau)= \begin{cases}
h(\tau), & \text{if } h(\tau)\leq m_1,\\ 
h(\tau)-m, & \text{if } h(\tau)\geq m_2.
\end{cases}
\end{equation*}
To define $h'(\xi)$ assume that $h(\xi)= (a,n_{0},R_{0},n_{1},\ldots, R_{k})$.
If the temporal extension of $h(\xi)$ contains no time point $> m_1$, then $h'(\xi)=h(\xi)$.
Otherwise, 
\begin{itemize}
\item[--] replace $n_0$ by $n_{0}-m$  if $n_{0}\geq m_2$, 
\item[--] replace all $n_i$ by $n_i -m$ if $\overline{n}_{i-1}\leq m_{1}$ and  $\overline{n}_i\geq m_{2}$,
\item[--] replace all $n_i$ by $n_i +m$ if  $\overline{n}_{i-1}\geq m_2$ and $\overline{n}_i\leq m_{1}$
\end{itemize}
and let $h'(\xi)$ be the resulting vector.
Using Lemmas~\ref{lem:abox}, \ref{lem:kk}, and \ref{l:anonym} one can readily check that
$h'$ is a homomorphism. 

After applying the above construction exhaustively, the resulting $h'$ is as required.
\end{proof}

\noindent
{\bf Theorem~\ref{thm:rewriting}}\ 
{\em Let $\T$ be a \TQL{} TBox in CoNF and $\q(\vec{x}, \vec{s})$ a totally ordered  CQ. Then, for any consistent KB $(\T,\A)$ and any tuples $\vec{a} \subseteq \ind(\A)$ and $\vec{n} \subseteq \mathbb Z$, 
$$
(\T, \A) \models \q(\vec{a},\vec{n}) \quad \text{iff} \quad \A^{\mathbb Z} \models \q^*(\vec{a},\vec{n}).
$$}
\begin{proof}
$(\Rightarrow)$ Suppose $(\T, \A) \models \q(\vec{a},\vec{n})$. Then, by Theorem~\ref{thm:canonical-bounded}, there is a homomorphism $g \colon \q(\vec{a},\vec{n}) \to \C^{\smash{d,\ell}}_\K$. If there is no individual variable $y$ such that $g(y) \notin \ind(\A)$, then we have $\C^{0}_\K \models \q(\vec{a},\vec{n})$, and so, by~\eqref{groundlevel}, $\A^{\mathbb Z} \models \ext^\T_{\q}(\vec{a},\vec{n})$, from which $\A^{\mathbb Z} \models \q^*(\vec{a},\vec{n})$ (just take $\mathfrak S = \emptyset$). 

Otherwise, we take a variable $y$ with $g(y) \notin \ind(\A)$ and consider the minimal set $\Theta_y$ of atoms in $\q$ with the following property: (i) all atoms containing $y$ are in $\Theta_y$; (ii) if $z$ is an individual variable in an atom from $\Theta_y$ and $g(z) \notin \ind(\A)$ then all atoms with $z$ are in $\Theta_y$. Consider the sub-query $\q' \subseteq \q$ comprised of all atoms in $\Theta_y$ (and all the temporal constraints in $\q$). Denote by $\mathcal{X}$ the set of individual terms in $\q'$ sent by $g$ to $\ind(\A)$, and by $\mathcal{Y}$ the remaining individual terms in $\q'$. Clearly, all elements of $\mathcal{Y}$ are existentially quantified variables. In view of the UNA, $\mathcal{X}$ contains at most one constant. By the construction of $\C_\K$ and $\Theta_y$, there are unique $b \in \ind(\A)$, $n \in \mathbb Z$ and a role $R$ such that $b \leadsto_R^n u$, for some $u$, and, for every $z \in \mathcal{Y}$, we have $b \leadsto_R^n u \leadsto^{n_1}_{R_1} \dots \leadsto^{n_m}_{R_m} g(z)$, for some $n_i$.

Consider now $\CKTR$ with root $a$. Let $e$ be the natural embedding of $\CKTR \setminus \{A(a,0)\}$ to $\C_\K$ such that $e(a) = b$, $e(0) = n$ and $e(R(a,v,0)) = R(b,u,n)$. Let $h \colon \q' \to \CKTR$ be such that $h \circ e = g$ on $\q'$. The map $h$ is not necessarily a tree witness for $R$ with $\q_h = \q'$: although $h(\q') \subseteq \CKTR^{\smash{d}}$, we may have $h(\q') \not\subseteq \CKTR^{\smash{d,\ell}}$. We use $h$ to construct a tree witness $h'$ for $R$ in the same way as in the proof of Theorem~\ref{thm:canonical-bounded}. 

Let $\mathfrak S$ be the set of all the distinct tree witnesses constructed in this way for the individual variables in $\q$. Observe that if $y,z \in \mathcal{Y}$, then $y$ and $z$ belong to the same tree witness. It follows that $\mathfrak S$ is consistent.
It follows now from~\eqref{groundlevel}, \eqref{existsR} 
and the proof of Theorem~\ref{thm:canonical-bounded} that 
\begin{equation*}
\A^{\mathbb Z} \models^g \bigwedge_{h' \in \mathfrak S} \mathsf{tw}_{h'}  \ \land \ \ext^\T_{\q \setminus \mathfrak S}. 
\end{equation*}

$(\Leftarrow)$ Suppose $A^{\mathbb Z} \models^{g'} \bigwedge_{h \in \mathfrak S} \mathsf{tw}_h  \ \land \ \ext^\T_{\q \setminus \mathfrak S}$ for some consistent set $\mathfrak S$ of tree witnesses and some assignment $g'$. Our aim is to extend $g'$ to an assignment $g$ in $\C_\K$ under which $\C_\K \models^g \q$. We set $g$ to coincide with $g'$ on those terms that occur in $\q^*$. Thus, it remains to define $g$ on the individual variables occurring in every $\q_h$, $h\in \mathfrak S$. Suppose $h$ is a tree witness for $R$, $g'(r_a) = n$ and $g'(x_h) = b$, where $r_a$ and $x_h$ are from $\tw_h$. By \eqref{existsR}, we have $\C_\K \models \exists R (b,n)$. Denote by $e$ the natural embedding of $\CKTR$ to $\C_\K$. Now, for every existentially quantified variable $y \in \mathcal{Y}_h$, we set $g(y) = h(e(y))$. Note that this definition is sound as different tree witnesses $h_1$ and $h_2$ in $\mathfrak S$ do not share variables apart form $\mathcal{X}_{h_1}\cap\mathcal{X}_{h_2}$, and the variables in the $\mathcal{Y}_{h_i}$ do not occur in $\q^*$. That the resulting $g$ is a homomorphism from $\q$ to $\C_\K$ follows from Lemmas~\ref{lem:kk} and~\ref{l:anonym} for the atoms in the tree witnesses $\q_h$ and from \eqref{groundlevel} and~\eqref{existsR} for the remaining atoms.
\end{proof}

%******************

\noindent
{\bf Theorem~\ref{thm:conphard}}
{\em Answering CQs over $\T = \{A \sqsubseteq \Diamond_{\FT}B\}$ is \coNP-hard for data complexity.}
\smallskip

\begin{proof}
The proof is by reduction of 2+2-SAT, a variant of propositional satisfiability that was first introduced by Schaerf as a tool for establishing lower bounds for the data complexity of query answering in a DL context~\cite{Schaerf-93}.  A \emph{2+2 clause} is of the form $(p_1 \vee p_2 \vee \neg n_1 \vee \neg n_2)$, where each of $p_1,p_2,n_1,n_2$ is a propositional variable or a truth constant $0$ and $1$. A \emph{2+2 formula} is a finite conjunction of 2+2 clauses. Now, 2+2-SAT is the problem of deciding whether a given 2+2 formula is satisfiable. 2+2-SAT is \NP-complete~\cite{Schaerf-93}.

Let $\varphi=c_0 \wedge \cdots \wedge c_{n}$ be a 2+2 formula in
propositional variables $\pi_0,\dots,\pi_{m}$, and let $c_i=p_{i,1} \vee
p_{i,2} \vee \neg n_{i,1} \vee \neg n_{i,2}$ for all $i \leq n$. 
Our aim is to define an ABox $\A_\varphi$ and a CQ $\q$ such that $\varphi$ is
unsatisfiable iff $(\T,\A_\varphi) \models \q$. We expand $\A$ and $\q_{1}\vee \q_{2}$ from Section~\ref{sec:non-rewritability}.
Namely, we represent the formula $\varphi$ in the ABox $\A_\varphi$ as follows:
\begin{itemize}
\item[--] the individual name $f$ represents the formula $\varphi$ and the individual names $c_0,\dots,c_n$ represent the clauses of $\varphi$;
\item[--] the assertions $S(f,c_0,0), \dots, S(f,c_{n},0)$ associate $f$
    with its clauses (at time point $0$), where $S$ is a role name;
\item[--] the individual names $\pi_0,\dots,\pi_m$ represent the propositional variables,
    and the individual names ${\sf false}$, ${\sf true}$ represent truth constants;
\item[--] the assertions, for each $i \leq n$,
\begin{equation*}
    P_1(c_i,p_{i,1},0),\ \ 
    P_2(c_i,p_{i,2},0), \ \ N_1(c_i,n_{i,1},0), \ \ N_2(c_i,n_{i,2},0)
\end{equation*}
associate each clause with the four variables/truth constants that occur in it (at time point $0$), where $P_1,P_2,N_1,N_2$ are role names.
\end{itemize}
We further extend $\A_\varphi$ to enforce a truth value for each of the variables $\pi_i$. To this end, add  $\A_0,\dots,\A_{m}$ to $\A_\varphi$, where $\A_i = \{ A(a^i,0), C(a^i,1) \}$. Intuitively, $\A_i$ is a copy of $\A$ from Section~\ref{sec:non-rewritability} and is used to generate a truth value for the variable $\pi_i$, where we want to interpret $\pi_i$ as true if $\q_{1}(a)$ is satisfied and as false if $\q_{2}(a)$ is satisfied. To actually relate each individual name $\pi_i$ to the associated ABox $\A_i$, we use role name $R$.
More specifically, we extend $\A_{\varphi}$ as follows:
\begin{itemize}
\item[--] link variables $\pi_i$ to the ABoxes $\A_i$ by adding assertions $R(\pi_i,a^{i},0)$ for all $i \leq m$;
thus, truth of $\pi_i$ means that ${\sf tt}(\pi_i)$ is satisfied and falsity means that ${\sf ff}(\pi_i)$ is satisfied, where
\begin{align*}
{\sf tt}(y) & = \exists z,t\, \bigl(R(y,z,0) \wedge C(z,t) \land B(z,t)\bigr),\\ %\q_{1}(x)),\\
{\sf ff}(y) & = \exists z,t,t' \,\bigl(R(y,z,0) \wedge  (t < t') \land C(z,t) \land B(z,t')\bigr); %\q_{2}(x));
\end{align*}

\item[--] to ensure that ${\sf false}$ and ${\sf true}$ have the expected truth values, add the following assertions to the ABox:
\begin{align*} 
&R({\sf true},{\sf true},0),A({\sf true},0),C({\sf true},1),B({\sf true},1),\\
&R({\sf false},{\sf false},0),A({\sf false},0),C({\sf false},1),B({\sf false},2).
\end{align*}
\end{itemize}
Consider the query
\begin{align*}
\q  =  \exists x\, \bigl(S(f,x,0) \ \ \wedge\ \  
              & (\exists y_{1} (P_{1}(x,y_{1},0) \wedge {\sf ff}(y_{1}))) \ \ \wedge \ \
              (\exists y_{2} (P_{2}(x,y_{2},0) \wedge {\sf ff}(y_{2})))\wedge {} \\
             & (\exists y_{3} (N_{1}(x,y_{3},0) \wedge {\sf tt}(y_{3}))) \ \ \wedge \ \
               (\exists y_{4} (N_{2}(x,y_{4},0) \wedge {\sf tt}(y_{4})))\bigr),
\end{align*}
which describes the existence of a clause with only false literals and thus captures falsity of $\varphi$. It is straightforward to show that $\varphi$ is unsatisfiable iff $(\T,\A_\varphi) \models \q$. To obtain the desired CQ, it remains to pull the existential quantifiers out.
\end{proof}

\end{document}